\documentclass[sigconf,nonacm]{acmart}
\renewcommand\footnotetextcopyrightpermission[1]{}
\newcommand{\nCorpusProjects}{500}
\newcommand{\nLoops}{18,707}
\newcommand{\nPureLoops}{2,510}
\newcommand{\pctPureLoops}{13.4\%}

\newcommand{\pctProjectsPureLoop}{59\%}
\newcommand{\nFreeLoops}{1,010}

\newcommand{\nPureLoopsInWarp}{516}

\newcommand{\nPureLoopsOutsideWarp}{1,994}

\newcommand{\nProjectsPureLoopOutsideWarp}{293}
\newcommand{\pctProjectsPureLoopOutsideWarp}{59\%}

\newcommand{\nProjectsFreeLoopOutsideWarp}{130}
\newcommand{\pctProjectsFreeLoopOutsideWarp}{26\%}

\newcommand{\nKindinputpolling}{71}
\newcommand{\nKindothercomputation}{1,412}
\newcommand{\nKindscorelives}{145}
\newcommand{\nKindlistwork}{402}
\newcommand{\nKindphysics}{224}
\newcommand{\nKindtimercounter}{243}
\newcommand{\nKindrandomAI}{13}

\newcommand{\nIdleMeasured}{499}
\newcommand{\nIdleHiddenDen}{477}
\newcommand{\nIdleHiddenDiverge}{81}
\newcommand{\pctIdleHiddenDiverge}{17\%}

\newcommand{\nIdleMutedDen}{477}
\newcommand{\nIdleMutedDiverge}{95}
\newcommand{\pctIdleMutedDiverge}{20\%}

\newcommand{\nIdleFasterDen}{499}
\newcommand{\nIdleFasterDiverge}{12}
\newcommand{\pctIdleFasterDiverge}{2\%}

\newcommand{\nIdleRawDen}{499}
\newcommand{\nIdleRawDiverge}{258}
\newcommand{\pctIdleRawDiverge}{52\%}

\newcommand{\nIdleRoundsOne}{362}

\newcommand{\nIdleRoundsAtBudget}{51}

\newcommand{\nIdleFree}{137}

\newcommand{\pctIdleLintBaseline}{27\%}
\newcommand{\nLintFlagFree}{77}
\newcommand{\nLintFlagStill}{216}
\newcommand{\nLintUnflagFree}{60}
\newcommand{\nLintUnflagStill}{146}
\newcommand{\pctLintRecall}{56\%}
\newcommand{\pctLintPrecision}{26\%}

\newcommand{\nDrivenMeasured}{499}
\newcommand{\nDrivenWithKeys}{353}
\newcommand{\nDrivenHiddenDen}{477}
\newcommand{\nDrivenHiddenDiverge}{84}
\newcommand{\pctDrivenHiddenDiverge}{18\%}

\newcommand{\nDrivenMutedDen}{477}
\newcommand{\nDrivenMutedDiverge}{116}
\newcommand{\pctDrivenMutedDiverge}{24\%}
\newcommand{\nDrivenMutedStage}{37}

\newcommand{\nDrivenFasterDen}{499}
\newcommand{\nDrivenFasterDiverge}{13}
\newcommand{\pctDrivenFasterDiverge}{3\%}

\newcommand{\nDrivenRawDen}{499}
\newcommand{\nDrivenRawDiverge}{305}
\newcommand{\pctDrivenRawDiverge}{61\%}
\newcommand{\nDrivenRawStage}{275}
\newcommand{\nDrivenRawVarsOnly}{30}
\newcommand{\nDrivenRawStageOnly}{93}

\newcommand{\nDrivenHiddenNotMuted}{31}
\newcommand{\nDrivenHiddenAndMuted}{53}
\newcommand{\nDrivenMutedNotHidden}{63}

\newcommand{\nDrivenFree}{132}

\newcommand{\movedMScore}{35}

\newcommand{\movedRScore}{45}

\newcommand{\nTappedMutedDen}{477}
\newcommand{\nTappedMutedDiverge}{117}
\newcommand{\pctTappedMutedDiverge}{25\%}

\newcommand{\nTappedRawDen}{499}
\newcommand{\nTappedRawDiverge}{312}
\newcommand{\pctTappedRawDiverge}{63\%}

\newcommand{\nBudgetSubset}{137}

\newcommand{\nBudgetDivergeSix}{14}

\newcommand{\nBudgetDivergeNinetyNine}{12}

\newcommand{\nBudgetDivergeFifteenNinetyNine}{12}

\newcommand{\nBudgetWitness}{137}
\newcommand{\nBudgetWitnessNever}{120}

\newcommand{\nRealProjects}{60}
\newcommand{\nRealAgreeOne}{50}
\newcommand{\pctRealAgreeOne}{83\%}
\newcommand{\nRealAgreeTwenty}{34}
\newcommand{\nRealAgreeBoth}{33}
\newcommand{\nRealDifferTwenty}{22}

\newcommand{\nRealDifferTwentyFree}{3}
\newcommand{\nRealDifferTwentyThrottled}{19}

\newcommand{\nFlipRemixes}{568}

\newcommand{\nFlipRemixesFish}{274}
\newcommand{\nFlipRemixesFly}{294}
\newcommand{\nFlipBasePass}{876}
\newcommand{\nFlipBaseFail}{244}
\newcommand{\nFlipBaseUntested}{584}
\newcommand{\nFlipRules}{1,704}
\newcommand{\nFlipB}{546}
\newcommand{\pctFlipB}{32\%}

\newcommand{\nFlipM}{45}

\newcommand{\nFlipRevM}{4}

\newcommand{\nFlipC}{0}

\newcommand{\nFlipR}{55}

\newcommand{\nFlipRevR}{4}

\newcommand{\nFlipRateBasePass}{950}
\newcommand{\nFlipRateBaseFail}{680}
\newcommand{\nFlipRateBaseUntested}{74}
\newcommand{\nFlipRateRules}{1,704}
\newcommand{\nFlipRateB}{652}

\newcommand{\nFlipRateM}{7}

\newcommand{\flipRateKindsM}{fail$\rightarrow$pass 4, pass$\rightarrow$fail 3}
\newcommand{\nFlipRateC}{0}

\newcommand{\nFlipRateR}{12}

\newcommand{\nFlipRateRevR}{9}

\newcommand{\nRealThrottledNoFreeFrame}{18}

\newcommand{\nFidAGames}{60}
\newcommand{\nFidAIdentical}{23}
\newcommand{\nFidAGate}{1}
\newcommand{\nFidAAmbiguous}{1}

\newcommand{\nFidASensing}{35}
\newcommand{\pctFidARoundsAgree}{99.4\%}
\newcommand{\nFidAClassAgree}{59}

\newcommand{\nFidBGames}{60}
\newcommand{\nFidBIdentical}{25}
\newcommand{\nFidBGate}{1}
\newcommand{\nFidBAmbiguous}{1}

\newcommand{\nFidBSensing}{33}
\newcommand{\pctFidBRoundsAgree}{99\%}
\newcommand{\nFidBClassAgree}{60}

\newcommand{\nFidMGames}{60}
\newcommand{\nFidMIdentical}{21}
\newcommand{\nFidMGate}{1}
\newcommand{\nFidMAmbiguous}{3}

\newcommand{\nFidMSensing}{35}
\newcommand{\pctFidMRoundsAgree}{96.3\%}
\newcommand{\nFidMClassAgree}{60}

\newcommand{\nFidCGames}{60}
\newcommand{\nFidCIdentical}{23}
\newcommand{\nFidCGate}{1}
\newcommand{\nFidCAmbiguous}{1}

\newcommand{\nFidCSensing}{35}
\newcommand{\pctFidCRoundsAgree}{99\%}
\newcommand{\nFidCClassAgree}{59}

\newcommand{\nFidRGames}{60}
\newcommand{\nFidRIdentical}{23}
\newcommand{\nFidRGate}{0}
\newcommand{\nFidRAmbiguous}{2}

\newcommand{\nFidRSensing}{35}
\newcommand{\pctFidRRoundsAgree}{95\%}
\newcommand{\nFidRClassAgree}{60}

\newcommand{\nWhiskerSuites}{16}
\newcommand{\nWhiskerTests}{119}
\newcommand{\nWhiskerBasePass}{103}
\newcommand{\nWhiskerBaseFail}{13}
\newcommand{\nWhiskerBaseSkip}{3}
\newcommand{\nWhiskerFlipsB}{53}

\newcommand{\nWhiskerRevB}{42}

\newcommand{\nWhiskerFlipsM}{20}
\newcommand{\pctWhiskerFlipsM}{17\%}
\newcommand{\nWhiskerRevM}{9}

\newcommand{\nWhiskerFlipsC}{1}

\newcommand{\nWhiskerFlipsR}{36}
\newcommand{\pctWhiskerFlipsR}{30\%}
\newcommand{\nWhiskerRevR}{25}

\newcommand{\nWhiskerNativeTests}{119}

\newcommand{\nWhiskerNativeAgreeA}{108}

\newcommand{\nWhiskerNativeAgreeW}{107}

\newcommand{\nBoatRemixes}{224}
\newcommand{\nBoatVerdicts}{2,016}

\newcommand{\nBoatFlipsB}{389}

\newcommand{\nBoatRevB}{353}

\newcommand{\nBoatFlipsM}{389}
\newcommand{\pctBoatFlipsM}{19\%}
\newcommand{\nBoatRevM}{361}

\newcommand{\nBoatFlipsC}{1}

\newcommand{\nBoatFlipsR}{885}
\newcommand{\pctBoatFlipsR}{44\%}
\newcommand{\nBoatRevR}{296}

\newcommand{\nBoatSpeedFlipsM}{111}

\newcommand{\nRandomAttempts}{2,141}
\newcommand{\nRandomOk}{270}

\newcommand{\randomMaxId}{1,390,000,000}
\newcommand{\nRandomProjects}{270}

\newcommand{\nRandomIdleMeasured}{263}
\newcommand{\nRandomIdleFree}{103}

\newcommand{\nRandomIdleBDen}{209}
\newcommand{\nRandomIdleBDiv}{23}
\newcommand{\pctRandomIdleBDiv}{11\%}
\newcommand{\nRandomIdleMDen}{209}
\newcommand{\nRandomIdleMDiv}{19}
\newcommand{\pctRandomIdleMDiv}{9\%}
\newcommand{\nRandomIdleCDen}{263}
\newcommand{\nRandomIdleCDiv}{7}
\newcommand{\pctRandomIdleCDiv}{3\%}
\newcommand{\nRandomIdleRDen}{262}
\newcommand{\nRandomIdleRDiv}{77}
\newcommand{\pctRandomIdleRDiv}{29\%}
\newcommand{\nRandomDrivenMeasured}{263}
\newcommand{\nRandomDrivenFree}{101}

\newcommand{\nRandomDrivenBDen}{209}
\newcommand{\nRandomDrivenBDiv}{25}
\newcommand{\pctRandomDrivenBDiv}{12\%}
\newcommand{\nRandomDrivenMDen}{209}
\newcommand{\nRandomDrivenMDiv}{24}
\newcommand{\pctRandomDrivenMDiv}{11\%}
\newcommand{\nRandomDrivenCDen}{263}
\newcommand{\nRandomDrivenCDiv}{8}
\newcommand{\pctRandomDrivenCDiv}{3\%}
\newcommand{\nRandomDrivenRDen}{262}
\newcommand{\nRandomDrivenRDiv}{87}
\newcommand{\pctRandomDrivenRDiv}{33\%}
\newcommand{\casePongTwoWebScoreA}{3}
\newcommand{\caseRunnerWebLivesA}{0}

\newcommand{\casePongTwoWebScoreB}{0}
\newcommand{\caseRunnerWebLivesB}{3}

\newcommand{\casePongTwoWebScoreM}{191}
\newcommand{\caseRunnerWebLivesM}{0}

\newcommand{\casePongTwoRoundsA}{1}

\newcommand{\casePongTwoRoundsM}{24}

\newcommand{\casePongTwoScoreA}{1}
\newcommand{\casePongTwoScoreM}{176}
\newcommand{\casePongTwoScoreB}{0}

\newcommand{\casePongRoundsB}{1}

\newcommand{\caseRunnerLivesA}{0}
\newcommand{\caseRunnerLivesB}{3}
\newcommand{\caseRunnerRoundsB}{21.5}
\newcommand{\caseRunnerRoundsM}{9.63}
\newcommand{\caseRunnerLivesM}{0}
\newcommand{\caseRunnerHidden}{4}

\newcommand{\caseFruitRoundsR}{10.59}
\newcommand{\caseFruitScoreA}{10}
\newcommand{\caseFruitScoreR}{82}

\newcommand{\editorVisibleSixty}{615}
\newcommand{\editorHiddenSixty}{24,353,468}

\newcommand{\editorVisibleThirty}{303}
\newcommand{\editorHiddenThirty}{23,422,395}
\newcommand{\editorRoundsPerFrameHiddenThirty}{77,000}
\newcommand{\editorHiddenThirtyCpuFour}{6,007,493}
\newcommand{\editorHiddenThirtyCpuTwenty}{1,143,273}
\newcommand{\editorVisibleThirtyCpuTwenty}{302}
\newcommand{\editorTurboVisibleSixty}{14,389,449}
\newcommand{\editorHiddenOverVisibleThirty}{77,000}

\newcommand{\editorMachineShort}{an Apple M5 Max laptop}
\newcommand{\harnessVisible}{300}
\newcommand{\harnessHidden}{7,200}

\newcommand{\harnessRoundsCap}{24}
\newcommand{\harnessRoundsFast}{99}

\newcommand{\vmVersionNote}{5.0.300}
\newcommand{\mutedOneIn}{four}

\newcommand{\nWhiskerFlipsMAgreed}{20}

\newcommand{\nWhiskerFlipsRAgreed}{34}

\newcommand{\nWhiskerAgreedTests}{108}

\newcommand{\ruleRevsM}{mf2 2, mf1 1, mf3 1}

\newcommand{\ruleRateRevsM}{mfr3 4, mfr1 3}

\newcommand{\nBoatRateFlipsM}{361}
\newcommand{\nBoatTriggerFlipsM}{28}

\newcommand{\nDrivenKeysMutedDiverge}{101}
\newcommand{\nDrivenKeysMutedDen}{341}
\newcommand{\pctDrivenKeysMutedDiverge}{30\%}

\newcommand{\nDrivenNoKeysMutedDiverge}{15}
\newcommand{\nDrivenNoKeysMutedDen}{136}
\newcommand{\pctDrivenNoKeysMutedDiverge}{11\%}

\newcommand{\pctDrivenBudgetNinetyNineMutedDiverge}{24\%}

\newcommand{\pctDrivenBudgetNinetyNineRawDiverge}{59\%}

\newcommand{\pctSeedTwoDiverge}{34\%}
\newcommand{\nWebDenB}{59}
\newcommand{\nWebAgreeB}{54}

\newcommand{\nWebDenM}{59}
\newcommand{\nWebAgreeM}{58}
\newcommand{\nWebBothM}{19}

\newcommand{\nWebDivergeM}{20}
\newcommand{\nWebDenR}{60}
\newcommand{\nWebAgreeR}{58}

\newcommand{\nProjectsRandomPureLoop}{25}
\newcommand{\nDrivenMutedDivergeRandomLoop}{6}

\newcommand{\nIdleFreeAboveTwo}{114}
\newcommand{\nIdleFreeAboveTen}{69}

\newcommand{\nDrivenProjRawDiverge}{253}
\newcommand{\nDrivenProjRawDen}{499}
\newcommand{\pctDrivenProjRawDiverge}{51\%}
\newcommand{\nDrivenHiddenReshown}{275}

\newcommand{\nSeedTwoHiddenPaired}{477}
\newcommand{\nSeedTwoHiddenLabelAgree}{466}

\newcommand{\nSeedTwoMutedDiverge}{117}
\newcommand{\nSeedTwoMutedDen}{477}
\newcommand{\pctSeedTwoMutedDiverge}{25\%}
\newcommand{\nSeedTwoMutedPaired}{477}
\newcommand{\nSeedTwoMutedLabelAgree}{470}

\newcommand{\nSeedTwoRawPaired}{499}
\newcommand{\nSeedTwoRawLabelAgree}{497}

\newcommand{\nRepeatPairs}{499}
\newcommand{\nRepeatIdentical}{499}
\newcommand{\nSingleGames}{477}
\newcommand{\nSingleGamesAny}{61}
\newcommand{\pctSingleGamesAny}{13\%}
\newcommand{\nSingleSprites}{1,950}
\newcommand{\nSingleSpritesDiverge}{62}

\newcommand{\pctDrivenBudgetThreeNinetyNineMutedDiverge}{23\%}

\newcommand{\pctDrivenBudgetThreeNinetyNineRawDiverge}{59\%}

\newcommand{\pctDrivenBudgetFifteenNinetyNineMutedDiverge}{23\%}

\newcommand{\pctDrivenBudgetFifteenNinetyNineRawDiverge}{58\%}
\newcommand{\nDrivenBudgetFifteenNinetyNineRawExcluded}{24}

\newcommand{\nBudgetRefPairs}{499}
\newcommand{\nBudgetRefSame}{489}

\newcommand{\nPerLoopNonParking}{11,264}

\newcommand{\nPerLoopFlagged}{1,496}
\newcommand{\nPerLoopFlaggedFree}{194}
\newcommand{\pctPerLoopPrecision}{13\%}
\newcommand{\pctPerLoopRecall}{17\%}
\newcommand{\pctPerLoopBase}{10\%}
\newcommand{\nPerLoopFlaggedEntered}{837}
\newcommand{\pctPerLoopPrecisionEntered}{23\%}

\newcommand{\nFidAPrefixFrames}{8,101}
\newcommand{\nFidAPrefixDisagree}{1}
\newcommand{\medFidAFirstState}{17}

\newcommand{\nRandomProjectsPureLoopOutsideWarp}{42}
\newcommand{\pctRandomProjectsPureLoopOutsideWarp}{16\%}
\newcommand{\nRandomProjectsFreeLoopOutsideWarp}{20}
\newcommand{\pctRandomProjectsFreeLoopOutsideWarp}{7\%}

\newcommand{\nGamesColour}{116}
\newcommand{\pctGamesColour}{23\%}
\newcommand{\nRandomColour}{21}
\newcommand{\nGamesBackdropLoop}{89}

\newcommand{\nDrivenNoColourMutedDiverge}{85}
\newcommand{\nDrivenNoColourMutedDen}{364}
\newcommand{\pctDrivenNoColourMutedDiverge}{23\%}
\newcommand{\nDrivenNoColourRawDiverge}{238}
\newcommand{\nDrivenNoColourRawDen}{383}
\newcommand{\pctDrivenNoColourRawDiverge}{62\%}

\newcommand{\nGamesRefCut}{1}

\newcommand{\nGamesNoSprite}{22}
\newcommand{\nRandomRefCut}{7}

\newcommand{\nRandomNoSprite}{54}

\newcommand{\ciRandomDrivenMutedDiverge}{8--17\%}

\newcommand{\nFidPrefixMulti}{1,059}
\newcommand{\nFidPrefixMultiAgree}{1,058}
\newcommand{\nFidPrefixMultiGames}{16}
\newcommand{\nFidFreeGames}{19}
\newcommand{\nFidFreeGamesAgree}{18}

\newcommand{\nRcmGames}{60}

\newcommand{\nRcmDiffersAll}{21}

\newcommand{\nRcmPaired}{59}
\newcommand{\nRcmAgree}{51}

\newcommand{\nRcmWebOnly}{5}
\newcommand{\nRcmHarnessOnly}{3}
\newcommand{\medRcmMutedRoundsPerFrame}{173}

\newcommand{\nEngineWebKitSame}{49}

\newcommand{\nEngineFirefoxSame}{48}

\newcommand{\nEngineNWebKitPairs}{60}
\newcommand{\nEngineNWebKitSame}{60}

\newcommand{\nBoatRateRevM}{348}
\newcommand{\nBoatWindowRevM}{13}

\newcommand{\nRemoteReferencePairs}{499}
\newcommand{\nRemoteReferenceSame}{499}

\newcommand{\nRemoteMutedSame}{499}

\newcommand{\nRemoteRawSame}{498}

\newcommand{\nCtRemixes}{571}
\newcommand{\nCtMoved}{239}
\newcommand{\pctCtMoved}{42\%}
\newcommand{\nCtFlipped}{34}
\newcommand{\nCtBoth}{32}
\newcommand{\nCtMovedNotFlipped}{207}
\newcommand{\nCtFlippedNotMoved}{2}

\newcommand{\editorVisibleThirtyMin}{303}

\newcommand{\editorHiddenThirtyMin}{23,416,569}
\newcommand{\editorHiddenThirtyMax}{24,928,746}

\usepackage{booktabs}
\usepackage{array}
\usepackage{xcolor}
\usepackage{algorithm}
\usepackage{algpseudocode}
\usepackage{tikz}
\usetikzlibrary{arrows.meta,positioning,shapes.geometric,backgrounds}
\newcommand{\runin}[1]{\smallskip\noindent\textbf{\emph{#1}}\ }
\newcommand{\takeaway}[1]{\smallskip\noindent\emph{#1}\smallskip}
\newcommand{\tool}{\textsc{ThrottleCheck}}
\newcommand{\Kone}{K1}
\newcommand{\Ktwo}{K2}
\newcommand{\Kthree}{K3}
\newcommand{\KoneM}{K1$'$}
\newcommand{\theSentence}{A loop that never moves a sprite, changes how one looks, draws with the pen, makes a clone, or waits runs as fast as your computer can, except while something else on the stage is being drawn; then it runs once per screen update.}
\begin{document}
\setlength{\emergencystretch}{1.5em}
\title[The Invisible Throttle]{The Invisible Throttle: Running on Borrowed Time in Scratch}
\author{Xinyue Feng}
\affiliation{\institution{Independent Researcher}\country{China}}
\email{lily.china@outlook.com}
\author{Hanyuan Shi}
\affiliation{\institution{Independent Researcher}\country{China}}
\email{shihanyuan1995@gmail.com}
\author{Yuan Si}
\authornote{Corresponding author.}
\affiliation{\institution{University of Waterloo}\country{Canada}}
\email{y3si@uwaterloo.ca}
\begin{abstract}
Scratch has 135 million registered users, most of them children, and 164 million shared projects. What they are taught about the speed of a script fits in one sentence: a loop iterates once per frame. Unfortunately, that sentence describes the exception. In the public virtual machine a frame repeats the scripts until something visible asks the screen to redraw, no script is left running, or three quarters of the frame's wall-clock time are spent. Hence a loop that does not draw is paced by whatever else is visible and by the machine. Hide the moving sprite of a two-sprite project, and the other's counting loop runs \editorHiddenOverVisibleThirty{} times faster on a laptop. No documentation states the rule, and a child who tunes a game tunes it to a rate the blocks never name.

Our key observation is that the redraw gate is one flag for the whole runtime, which implies that a loop's speed can be changed without touching its code: hide the sprite that draws (\Kone{}), change the machine's budget (\Ktwo{}), or run the program in a tool with no renderer (\Kthree{}). \tool{} implements a budgeted semantics (a stated budget of rounds per frame, the redraw gate emulated) on the unmodified virtual machine, runs a project under each knob, and reports a rate-sensitive project with a witness; a fourth knob, \KoneM{}, mutes a sprite's requests without hiding it and so isolates the throttle.

On \nCorpusProjects{} popular public games, \pctProjectsPureLoop{} contain a loop that never draws and never waits. Muting the requests of the sprites that draw changes the state of a played game after ten seconds in \pctDrivenMutedDiverge{} of the games that have one; hiding the same sprites changes \pctDrivenHiddenDiverge{}, and in part different ones, since a hidden sprite also leaves the game; running the game without a renderer and without the gate changes \pctDrivenRawDiverge{}. A random sample of \nRandomProjects{} gives \pctRandomDrivenBDiv{}, \pctRandomDrivenMDiv{} and \pctRandomDrivenRDiv{}; the budget alone moves \pctDrivenFasterDiverge{}. Rules that check for an event within a window of frames rarely reverse; rules that read a position, a score or a clock after a fixed time reverse between pass and fail once the throttle is released, in \pctWhiskerFlipsM{} of Whisker's own example tests and \pctBoatFlipsM{} of a tutorial's checks on \nBoatRemixes{} of its remixes. The emulated gate runs the same rounds as the editor's in \pctFidARoundsAgree{} of frames, and the editor's bundles under the harness's clock and seed agree with it on which games a knob moves (\nWebAgreeM{} of \nWebDenM{} under muting). A tool that grades Scratch programs without a renderer grades a program the editor never runs, unless it emulates the gate and states a budget; we say what a grader should report, what the platform could change, and close with the sentence a child could be taught.
\end{abstract}
\maketitle

\section{Introduction}\label{sec:intro}
Scratch counted 164 million shared projects and 135 million registered users in July 2024~\cite{scratchstats}, most of them children, and for many of them it is the first programming language they meet~\cite{resnick2009scratch,maloney10}. Its blocks are designed to show everything the computer will do~\cite{resnick2009scratch,bau2017learnable}. This paper is about a clock the blocks do not show.

Under the surface a Scratch project is a concurrent program~\cite{maloney10}: a script for the cat, a script for the apple, a script for the score, each run by the virtual machine as a cooperative thread over shared state. What every learner is told about the speed of these threads fits in one sentence. Loops, says the Scratch Wiki, ``by default iterate once per frame''~\cite{wikisingleframe}; at the end of each cycle, a forum regular explains, ``the loop will wait for the screen to get refreshed before running again''~\cite{scratchvm2138}. Thirty frames a second, one iteration a frame, the same on every computer. The tools that run Scratch programs outside the editor~\cite{johnson2016itch,stahlbauer2020verified,stahlbauer19,deiner23whisker,goetz2022model,feldmeier24} describe no other model (Section~\ref{sec:related}).

Unfortunately, the sentence describes the exception. We read the scheduler of the public Scratch~3 virtual machine (\texttt{sequencer.js}, version \vmVersionNote{}) and found that a frame is not one round of the scripts. It is as many rounds as fit before one of three exits: some primitive has requested a redraw (E1), no script is left running (E2), or three quarters of the frame's wall-clock budget, 25\,ms of 33, are spent (E3). A redraw is requested by a visible sprite that moves, turns or changes its looks, by its speech bubble, by the pen, and by a \emph{wait} block. It is never requested by a variable, a list, an operator or a sensing block, and never by a hidden sprite's motion or looks. So a loop that never draws is held to one iteration per frame only while some \emph{other} visible sprite is drawing. When nothing visible changes, it runs as many iterations per frame as the computer manages in 25 milliseconds, tens of thousands on a laptop.

The rule has been in the source since at least 2019. A bug report from that year describes one loop starving the rest of a project, an expert replies that the link between ``refresh'' and loops is ``a fairly common misunderstanding''~\cite{scratchvm2138}, and the report is still open. The Wiki's page on the \emph{forever} block says that the loop waits for the next frame ``provided there are blocks running in the project that require a yield''~\cite{wikiforever}. The Wiki's page on running faster says only that ``some blocks will yield'' and names the sound volume and effect blocks~\cite{wikisingleframe}. We know of no description, on the Wiki, in the documentation or in the literature on Scratch tools (Section~\ref{sec:related}), that says which blocks those are, that a hidden sprite's do not count, or that a wall-clock budget is the third exit. The 500\,ms limit the Wiki documents for \emph{run without screen refresh} is the warp timer, a clock inside one block, not the frame's.

\runin{Running example} Figure~\ref{fig:example} shows the smallest project that has it: a cat that counts (\emph{set n to 0}, then \emph{forever change n by 1}) and an apple that moves (\emph{forever move 1 step, if on edge bounce}). Nothing connects them. In the public editor the cat counts \editorVisibleThirty{} in ten seconds while the apple is visible and \editorHiddenThirty{} once the apple is hidden, \editorHiddenOverVisibleThirty{} times more, from a change to a sprite the cat's script never mentions.

Three properties make this more than a curiosity. The rule is \textbf{invisible}: no block reads or sets the rate, and Scratch's learners, who program bottom-up and by trial until the stage looks right~\cite{meerbaum2011habits}, tune to a rate they cannot name. It is \textbf{global}: the redraw flag is one flag for the whole runtime, so a visible sprite that moves anywhere on the stage ends the frame for every thread in the project. It is \textbf{environmental}: the clock exit is measured in wall-clock time, so the same project counts differently on a fast computer and a slow one, at sixty frames per second, and in a tool that has no screen at all. For a child the consequence is a game that plays differently on a friend's computer and, once a tool grades it, a verdict that depends on the tool rather than on the program.

Three findings surprised us. A slower or faster machine reaches most games not through the throttle but through their own clocks: of the \nRealDifferTwenty{} real-clock games that end elsewhere on a machine twenty times slower, \nRealDifferTwentyThrottled{} never free-run at all (RQ3). Hiding a sprite and silencing its requests are two interventions, not one: of the played games that hiding changes, a third change only because a hidden sprite touches nothing (RQ2). And whether a rubric survives depends on what it reads: rules of the form \emph{whenever A, then B within k frames} reverse \nFlipRevM{} times in \nFlipRules{} verdicts under the throttle, while one rule about a boat's pace flips in \nBoatSpeedFlipsM{} of \nBoatRemixes{} remixes (RQ5).

Our key observation is that the redraw gate is one flag for the whole runtime, which implies that the rate of a non-drawing loop is a function of what else is on the stage and of the machine, and not of the loop. This gives the paper its skeleton: three knobs, each a change to the environment and none to the code. \Kone{} hides the sprite that draws, and a loop elsewhere in the project speeds up by orders of magnitude. A game's timer, physics or scoring then behaves differently on a title screen, on a pause screen, or after a remixer hides a decorative sprite. \Ktwo{} changes the computer, and the same loop counts a different number of times per frame. \Kthree{} removes the renderer, so that no sprite primitive ever requests a redraw and every loop that does not wait free-runs. Hence a tool that executes a Scratch program without a renderer, and does not emulate the gate, runs a program the editor never runs.

\noindent This paper makes the following contributions.
\begin{enumerate}
  \item We state the rule from the virtual machine's source as a model with three exits and a redraw predicate tabulated by primitive, renderer and visibility, and derive the three knobs in it. We show it in the public editor, where hiding one sprite makes another's loop count \editorHiddenOverVisibleThirty{} times higher (Section~\ref{sec:background}).
  \item We describe \tool{}: a budgeted execution semantics (a stated budget of rounds per frame, the redraw gate emulated on the unmodified virtual machine) and a metamorphic check whose fourth knob mutes a sprite's requests without hiding it and so isolates the throttle. It reports a witness for every rate-sensitive project and carries a lint that writes one sentence next to the loops it is about (Section~\ref{sec:tool}).
  \item We measure the rule on \nCorpusProjects{} public games and a random sample of \nRandomProjects{} public projects, in the editor's own runtime on a real clock, and frame by frame against its gate. We also measure what it does to grading: \nFlipRemixes{} tutorial remixes under rules of two kinds, Whisker's own \nWhiskerTests{} example tests run unchanged, and a tutorial's steps over \nBoatRemixes{} of its remixes (Section~\ref{sec:eval}).
  \item We walk one real game and one control through the knobs and say what a learner, a grader and the platform can do: one sentence next to the loop, a third outcome beside pass and fail, and a pacing rule with an opt-out (Section~\ref{sec:cases}).
\end{enumerate}

\begin{figure}[t]\centering
  \begin{minipage}[t]{0.44\columnwidth}\centering\small (a) Cat\\[2pt]\includegraphics[width=0.92\linewidth]{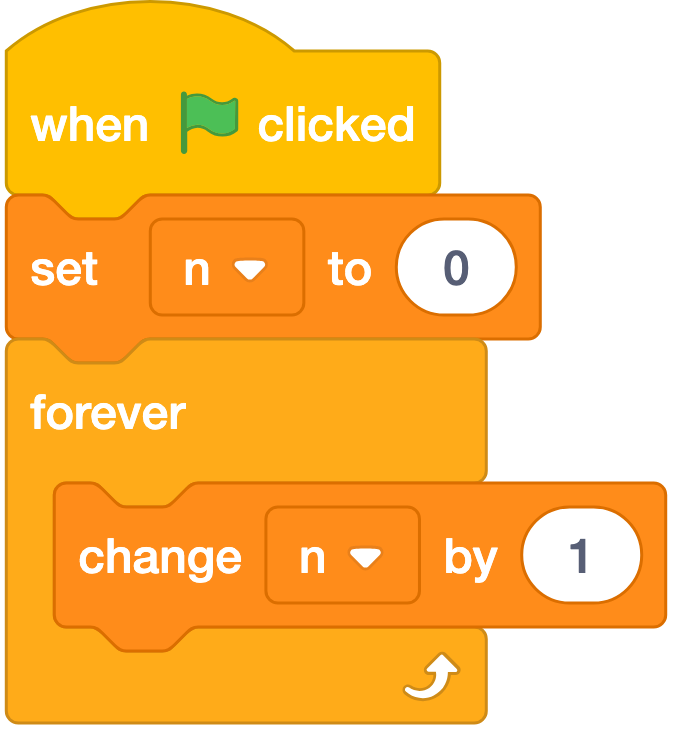}\end{minipage}\hfill
  \begin{minipage}[t]{0.44\columnwidth}\centering\small (b) Apple\\[2pt]\includegraphics[width=0.78\linewidth]{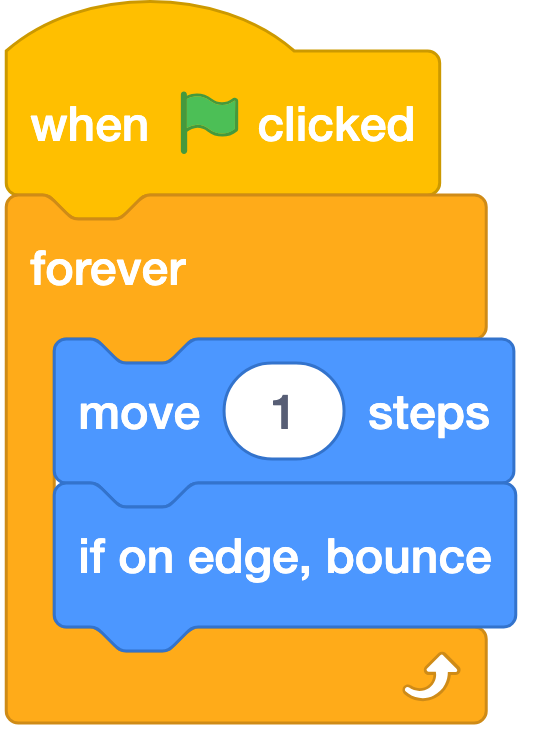}\end{minipage}\\[6pt]
  \small\setlength{\tabcolsep}{5pt}\begin{tabular}{@{}lr@{}}\toprule
  In the editor, after 10 seconds & $n$ \\ \midrule
  Apple visible, 30 frames per second & \editorVisibleThirty{} \\
  Apple visible, 60 frames per second & \editorVisibleSixty{} \\
  Apple visible, CPU 20$\times$ slower & \editorVisibleThirtyCpuTwenty{} \\
  Apple \textbf{hidden}, 30 frames per second & \editorHiddenThirty{} \\
  Apple hidden, CPU 4$\times$ slower & \editorHiddenThirtyCpuFour{} \\
  Apple hidden, CPU 20$\times$ slower & \editorHiddenThirtyCpuTwenty{} \\
  Apple hidden, 60 frames per second & \editorHiddenSixty{} \\
  Apple visible, turbo mode, 60 frames per second & \editorTurboVisibleSixty{} \\
  \bottomrule\end{tabular}
  \caption{The invisible throttle in one project. The cat counts (a) and the apple moves (b); nothing connects them. With the apple visible the cat counts once per frame, whatever the machine; with the apple hidden it counts as fast as the machine allows, about \editorHiddenOverVisibleThirty{} times faster on \editorMachineShort{}, and twenty times slower on a machine twenty times slower. Frame rate moves the first number, machine speed the second; turbo mode moves the first to the second's regime (lower than the hidden case because the apple's motion and edge test then run in every round). Each number is the median of three runs; at 30 frames per second all three give \editorVisibleThirtyMin{} with the apple visible and span \editorHiddenThirtyMin{}--\editorHiddenThirtyMax{} with it hidden.}
  \label{fig:example}
\end{figure}

\section{The Frame and Its Three Exits}\label{sec:background}
The model below is the frame of \texttt{sequencer.js} with names, not a new semantics of Scratch.

\subsection{How Scratch runs a project}\label{sec:vm}
A project is a set of sprites and a stage, each carrying scripts. A script starts with a \emph{hat} (\emph{when green flag clicked}, \emph{when this sprite clicked}, \emph{when I receive}); when the hat fires, the script becomes a \emph{thread}. Threads share everything the project owns: variables and lists, and the position, direction, size, costume and visibility of every sprite. The virtual machine runs them cooperatively, one at a time, and switches between them only where a thread stops: at the end of each iteration of a loop, at a block that parks the thread (\emph{wait}, \emph{wait until}, \emph{glide}, \emph{say for}, \emph{think for}, \emph{ask and wait}, \emph{broadcast and wait}, \emph{play sound until done}), and at the end of the script.

\runin{The frame} Figure~\ref{fig:mechanism} shows the scheduler's frame. Every 33\,ms the runtime calls the sequencer, which steps each running thread until it stops; one pass over all threads is a \emph{round}. The sequencer repeats the round while none of three exits has fired: (E1) some primitive has requested a redraw (turbo mode disables this exit), (E2) no thread is still running, and (E3) three quarters of the frame's wall-clock time have elapsed. The source tests them in the order E2, E3, E1. Then the stage is drawn and the runtime waits for the next frame.

\runin{Who requests a redraw} The request is a flag set by the primitives that change what the screen shows, and Table~\ref{tab:draw} gives the predicate as the source defines it. The eight sprite methods of the table request only if the sprite is visible after the call, so a hidden sprite's motion and looks never do, and neither does a change of layer, which only reorders the drawables. A visible sprite's bubble, a visible clone's creation and deletion, and the pen request as the table says, the pen whether or not its sprite is visible. The \emph{wait} block requests too, once, when it starts, which is why a loop containing a wait paces itself. Of the other parking blocks, \emph{glide} requests as motion does, \emph{say for}, \emph{think for} and \emph{ask} as bubbles do, and \emph{wait until}, \emph{broadcast and wait} and \emph{play sound until done} not at all. A variable's monitor never sets the flag either; the editor redraws it on its own. All of the sprite, bubble and pen requests sit behind a test that a renderer is attached. A virtual machine that runs without one never sees them, and only \emph{wait} still requests.

\begin{table}[t]\centering\footnotesize
\caption{The redraw predicate, read from the virtual machine (\texttt{rendered-target.js} and the looks, pen and control primitives).}\label{tab:draw}
\setlength{\tabcolsep}{4pt}
\begin{tabular}{@{}>{\raggedright\arraybackslash}p{0.46\columnwidth}>{\raggedright\arraybackslash}p{0.50\columnwidth}@{}}\toprule
Primitive & Requests a redraw \\ \midrule
Motion and looks of a sprite: position, direction, show, size, costume, effect, clear effects, rotation style & with a renderer attached, if the sprite is visible after the call \\ \addlinespace[2pt]
Backdrop switch & with a renderer attached; the stage is always visible \\ \addlinespace[2pt]
Speech or thought bubble, non-empty (\emph{ask}'s question included) & with a renderer attached, if the sprite is visible \\ \addlinespace[2pt]
A clone created (below the 300-clone cap) or deleted & with a renderer attached, if the clone is visible \\ \addlinespace[2pt]
Pen down; a pen-down sprite moving; stamp; clear & with a renderer attached, whether or not the sprite is visible \\ \addlinespace[2pt]
\emph{wait}, once, when its timer starts & always, renderer or not \\ \addlinespace[2pt]
Change of layer; pen up; data, operators, sensing, sound, events and the other control blocks & never \\
\bottomrule\end{tabular}
\end{table}

\subsection{The example in the editor}\label{sec:example}
The cat's script (Figure~\ref{fig:example}a) and the apple's (Figure~\ref{fig:example}b) share no variable, no broadcast and no contact. We ran the project for ten seconds in a browser hosting the public bundles with a real clock, the way the editor runs it (Section~\ref{sec:study}). With the apple visible, $n$ reaches \editorVisibleThirty{}: one count per frame. Hide the apple before pressing the flag, and $n$ reaches \editorHiddenThirty{}, about \editorRoundsPerFrameHiddenThirty{} counts per frame. The apple still moves and the cat still counts. The only change is that the apple's motion no longer reaches the screen, so it no longer ends the frame for the cat.

The table in Figure~\ref{fig:example} adds the two other knobs. Throttling the processor twenty times leaves the visible case at \editorVisibleThirtyCpuTwenty{} and cuts the hidden case to \editorHiddenThirtyCpuTwenty{}. Sixty frames per second, as the TurboWarp player allows~\cite{turbowarpfps}, doubles the visible count to \editorVisibleSixty{} and leaves the hidden count within the run-to-run noise of a real clock. The visible cat is paced by the screen and does not see the machine; the hidden cat is paced by the machine and does not see the screen. Hence the rate of a loop is not a property of the loop.

\subsection{The model}\label{sec:model}
A thread is \emph{stepped} by executing its blocks from where it last stopped until it stops again. After its step it is \emph{running} if it stopped at the end of a loop iteration, \emph{parked} at a parking block whose condition is unmet, and \emph{done} if its script ended. A \emph{round} steps, in the sprites' stacking order~\cite{schedcheck26}, every thread that is running or parked (a parked thread re-checks its condition and stays parked or resumes). A thread that a broadcast or a clone starts during a round is appended and stepped in the same round; a running script that a broadcast restarts keeps its place. Only running threads keep a frame going: a frame whose threads are all parked runs one round and exits by E2. A round is a function of the project state, the input, the clock and the random source. Let $\mathit{draw}(r)$ hold when some primitive executed during round $r$ requested a redraw (Table~\ref{tab:draw}).

\begin{definition}[Frame]\label{def:frame}
Under the editor's semantics $\mathcal{E}$, a frame executes rounds $r_1, \ldots, r_n$, where $n \geq 0$ is the least index at which one of three exits fires after $r_n$: \emph{E1}, $\mathit{draw}(r_n)$ holds; \emph{E2}, no thread is running; \emph{E3}, the wall-clock time consumed by $r_1 \ldots r_n$ has reached three quarters of the frame period. A frame with no thread executes no round ($n = 0$); a frame whose threads are all parked executes one round, in which they re-check their conditions, and exits by E2. E1 and E3 cannot fire before $r_1$, since the flag is cleared and the timer started when the frame starts. Turbo mode removes E1, and inside a \emph{warp} block (\emph{run without screen refresh}) a loop's end does not stop the thread, which yields only after 500\,ms of wall time; we model the default mode.
\end{definition}

\begin{proposition}[Throttle and free run]\label{prop:throttle}\label{prop:free}
If some primitive requests a redraw in the first round of a frame, the frame consists of exactly one round, and every thread running at its start and not stopped by another thread in that round is stepped exactly once. If no primitive requests a redraw in any round and some thread is running after every round, the frame consists of as many rounds as the machine starts within three quarters of the period, a number that depends on the machine and on the work the rounds do.
\end{proposition}
\begin{proof}
In the first case $\mathit{draw}(r_1)$ holds, so $n = 1$. In the second E1 and E2 never fire, so $n$ is set by E3 alone.
\end{proof}

We call the number of rounds of the second case the \emph{free-run count} of that frame on that machine. The first case is the coupling and the second the machine dependence; together they give the three knobs. Below, \emph{muting} a sprite means suppressing its redraw requests while it stays visible and touchable (Section~\ref{sec:tool} says how).

\begin{proposition}[The three knobs]\label{prop:knobs}
Let $t$ be a thread that executes no redraw-requesting primitive and is running throughout the frames considered, and compare two executions that start a frame from the same state and whose first rounds read the same input, clock values, random draws and sensing answers (the later rounds of the longer execution have no counterpart). For \KoneM{}, let $s$ be a sprite whose primitives, its clones' included, make the only redraw requests of the first round. With those requests muted, E1 does not fire after the first round, so $t$ is stepped a second time unless E2 or E3 fires first, and the free-run count times when no primitive requests in the later rounds either. For \Kone{}, hiding $s$ removes its motion, looks and bubble requests but not its pen's (Table~\ref{tab:draw}) and changes what sensing answers, so the clause holds when $s$ draws no pen line and while no read differs. For \Ktwo{}, two budgets, a slower and a faster machine that read the same clock values, run the same rounds in every frame that E1 or E2 ends within the smaller budget, and differ only in the number of rounds of the frames that E3 ends under it. For \Kthree{}, with no renderer attached only \emph{wait} requests a redraw, so in every frame in which no thread reaches a \emph{wait}, $t$ is stepped the free-run count times.
\end{proposition}
\begin{proof}
For \KoneM{}, with the requests of $s$ muted the first round makes no request, so E1 does not fire before $r_2$. Proposition~\ref{prop:throttle} applies to the later rounds when $\mathit{draw}$ stays false and some thread stays running, and only then: a sprite that moves once a counter reaches a value can request in the second round. For \Ktwo{}, the machine enters Definition~\ref{def:frame} only through E3: if E1 or E2 fires after some round $n$ that both machines start, both run exactly $r_1 \ldots r_n$, and otherwise each runs as many rounds as it starts. For \Kthree{}, every requester but \emph{wait} is behind the renderer test (Table~\ref{tab:draw}), so E1 fires only where a \emph{wait} starts.
\end{proof}

The proposition is about one frame; the study applies the knobs over 300 frames, where a released loop also consumes the random stream faster, which Section~\ref{sec:threats} bounds.

\begin{definition}[Budgeted semantics $\mathcal{C}_B$]\label{def:canonical}
For a budget $B \geq 1$ and a seed, a frame executes rounds $r_1, \ldots, r_n$ where $n \geq 0$ is the least index at which E1 or E2 fires after $r_n$, or $n = B$; a frame with no thread executes no round and an all-parked frame one, as under $\mathcal{E}$. $\mathit{draw}$ is evaluated as in $\mathcal{E}$ with a renderer attached. The clock read by the \emph{timer}, \emph{days since 2000} and \emph{current time} blocks advances one frame period per frame and is constant within a frame, and the timers behind \emph{wait}, \emph{glide}, \emph{say for} and the other parking blocks read the same clock. Random draws come from a generator started at the seed. A \emph{warp} block (\emph{run without screen refresh}) yields after 500 checks of its limit rather than after 500\,ms, a count independent of $B$.
\end{definition}

\begin{proposition}[Determinism and agreement]\label{prop:agree}\label{prop:canonical}
Under $\mathcal{C}_B$ the state after every frame is a function of the project, its input and the seed, and the same on every machine that runs the same JavaScript engine (floating-point arithmetic being the engine's). Consider a frame that $\mathcal{E}$ and $\mathcal{C}_B$ start from the same state with \emph{aligned reads}: the rounds read the same input, clock values and random draws, receive the same answers from sensing, and yield from \emph{warp} at the same points. If $\mathcal{E}$ ends the frame by E1 or E2 after $n \leq B$ rounds, $\mathcal{C}_B$ runs the same rounds and ends by the same exit. Hence a project every frame of which $\mathcal{E}$ ends by E1 or E2 within $B$ rounds reaches, with aligned reads throughout, the same states under $\mathcal{E}$ on any machine and under $\mathcal{C}_B$.
\end{proposition}
\begin{proof}
Determinism is by induction on frames: each round is a function of the state, the input, the clock (constant within the frame) and the generator's state, and the exits are functions of the state after the round or of the count. Agreement is by induction on the rounds: before round $k \leq n$ neither semantics has fired an exit ($\mathcal{C}_B$'s third fires only when $k - 1 = B$, which $k \leq n \leq B$ excludes), so round $k$ is the same under both, and after $r_n$ the same exit fires under both.
\end{proof}

Once their reads are aligned, the two semantics can disagree only in frames that the clock ends, and only in how many rounds run: that is what a stated $B$ standardizes. The alignment is a real restriction. Two machines do not read the same clock values, and Section~\ref{sec:eval} finds that the machine reaches most of the games it changes through their \emph{timer}, \emph{wait} and \emph{glide} blocks, a channel that only the virtual clock of $\mathcal{C}_B$ removes. Sensing answers differ between the harness and the editor in \nFidASensing{} of \nFidAGames{} games within 300 frames (Table~\ref{tab:fidelity}). Determinism makes the corpus study possible, since divergence under a knob is then a property of the project. It also says what we ask of a tool that runs Scratch programs without a renderer: implement $\mathit{draw}$ as the editor does and state its $B$.

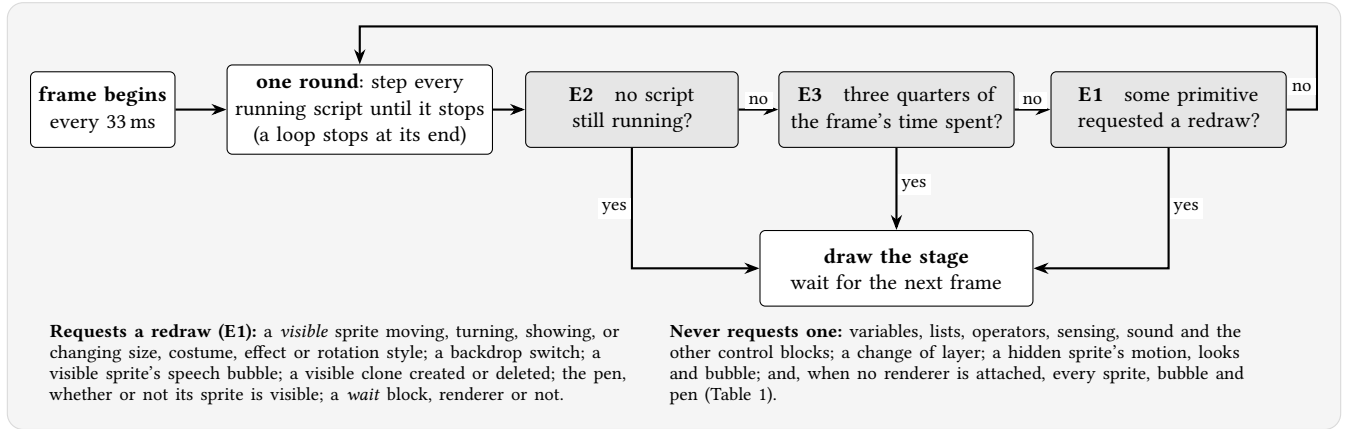
\begin{figure*}[t]\centering
\begin{tikzpicture}[x=1mm,y=1mm,font=\small,
  show background rectangle, background rectangle/.style={fill=black!4,draw=black!15,rounded corners=6pt}, inner frame sep=3mm,
  box/.style={draw,rounded corners=2pt,align=center,inner sep=3pt,minimum height=10mm,fill=white},
  test/.style={draw,rounded corners=2pt,align=center,inner sep=3pt,minimum height=10mm,fill=black!10},
  arr/.style={-{Stealth[length=2mm]},thick},
  lab/.style={font=\footnotesize,fill=white,inner sep=1.5pt}]
  \node[box,text width=17mm] (frame) at (9,0) {\textbf{frame begins}\\every 33\,ms};
  \node[box,text width=33mm] (round) at (43,0) {\textbf{one round}: step every\\running script until it stops\\(a loop stops at its end)};
  \node[test,text width=26mm] (idle) at (79,0) {\textbf{E2}\quad no script\\still running?};
  \node[test,text width=29mm] (budget) at (114,0) {\textbf{E3}\quad three quarters of\\the frame's time spent?};
  \node[test,text width=29mm] (redraw) at (150,0) {\textbf{E1}\quad some primitive\\requested a redraw?};
  \node[box,text width=34mm] (draw) at (114,-21) {\textbf{draw the stage}\\wait for the next frame};
  \draw[arr] (frame) -- (round);
  \draw[arr] (round) -- (idle);
  \draw[arr] (idle) -- node[lab,above]{no} (budget);
  \draw[arr] (budget) -- node[lab,above]{no} (redraw);
  \draw[arr] (redraw.east) -- ++(4,0) |- (43,11) -- (round.north);
  \node[lab] at ([shift={(2,3)}]redraw.east) {no};
  \draw[arr] (idle.south) |- node[lab,left,pos=0.25]{yes} (draw.west);
  \draw[arr] (budget.south) -- node[lab,right,pos=0.45]{yes} (draw.north);
  \draw[arr] (redraw.south) |- node[lab,right,pos=0.25]{yes} (draw.east);
  \node[anchor=north,text width=158mm,inner sep=0pt] at (81,-28) {\footnotesize\begin{tabular}{@{}p{76mm}@{\hspace{6mm}}p{76mm}@{}}
\textbf{Requests a redraw (E1):} a \emph{visible} sprite moving, turning, showing, or changing size, costume, effect or rotation style; a backdrop switch; a visible sprite's speech bubble; a visible clone created or deleted; the pen, whether or not its sprite is visible; a \emph{wait} block, renderer or not. &
\textbf{Never requests one:} variables, lists, operators, sensing, sound and the other control blocks; a change of layer; a hidden sprite's motion, looks and bubble; and, when no renderer is attached, every sprite, bubble and pen (Table~\ref{tab:draw}).
\end{tabular}};
\end{tikzpicture}
\caption{The per-frame work loop of the Scratch virtual machine (\texttt{sequencer.js}). Rounds repeat until no script is running (E2), the wall-clock budget is spent (E3) or a redraw was requested (E1), so one visible sprite's motion anywhere on the stage ends the frame for every loop in the project.}
\label{fig:mechanism}
\end{figure*}

\section{\tool{}}\label{sec:tool}
\tool{} implements the budgeted semantics on the unmodified virtual machine, turns the knobs into perturbations, reports a witness when a knob changes a project, and annotates the loops.

\runin{A budgeted runtime} The harness loads a project into \texttt{scratch-vm} with no renderer attached and steps it frame by frame under its own control, with three substitutions. First, \emph{a virtual clock and a seed}: the \emph{timer}, \emph{days since 2000} and \emph{current time} blocks read a clock that advances one frame period per frame, and the random source is seeded (seed 1 throughout). No audio engine is attached, so \emph{play sound until done} returns at once, and \emph{ask and wait} parks its thread until the driver answers: no block resolves on the host's clock. Second, \emph{a round budget}: the sequencer's own timer advances a fixed $\delta$ per read. The frame starts with one read and every exit check makes one, so a frame that nothing else ends runs $\lceil 25\,\mathrm{ms}/\delta \rceil - 1$ rounds, \harnessRoundsCap{} at the default $\delta = 1$\,ms and \harnessRoundsFast{} at $0.25$\,ms. A \emph{warp} block yields after 500 checks. Third, \emph{an emulated redraw gate}: we wrap the eight methods of Table~\ref{tab:draw} (the costume method on the stage is the backdrop switch) so that they request a redraw under the same visibility condition. We do the same for a visible sprite's non-empty bubble, a visible clone's creation and deletion, and the pen on pen down, on the motion of a pen-down sprite, on stamp and on clear (Table~\ref{tab:draw}); \emph{wait} requests natively. Section~\ref{sec:eval} checks this emulation against the editor's own gate frame by frame. Under this semantics the example project counts \harnessVisible{} in 300 frames with the apple visible and \harnessHidden{} with it hidden, on every machine (Proposition~\ref{prop:agree}). The renderer is not only the gate; sensing reads it. The harness answers sensing with a geometry of bounding boxes: \emph{touching} is an overlap of boxes, a click hits the topmost visible sprite whose box contains the point, and the colour predicates are answered false. The answers are the same under every knob but hiding, which changes what a touch finds, so a knob changes nothing but the gate, the budget and the visibility of the chosen sprites. Against the editor the geometry is a second difference, which Section~\ref{sec:threats} keeps apart from the gate.

\runin{The knobs as perturbations} Each knob is a run that differs from the reference in one setting, applied before the green flag and never to the program's blocks. Under \Kone{}, \emph{drawing sprites hidden}, every visible sprite whose scripts contain a motion, looks or pen block is hidden before the flag, as a \emph{hide} block would hide it. Its motion, looks and bubble stop requesting, its pen keeps requesting, a script that runs \emph{show} re-shows it (in \nDrivenHiddenReshown{} of the \nDrivenHiddenDen{} played games some hidden sprite is visible again by frame 300), and a clone inherits its visibility. The stage, which cannot be hidden, is neither hidden nor muted, so a loop that switches the backdrop (\nGamesBackdropLoop{} games have one) keeps requesting under both knobs. Hiding has two effects in Scratch, the release of the loops the sprite paced and the ordinary semantics of a hidden sprite, which neither touches nor is clicked. Hence we also run \KoneM{}, \emph{drawing sprites muted}: the same sprites stay visible, touchable and clickable, and only their redraw requests (methods, bubble and pen, for the sprite and its clones) are suppressed. \KoneM{} isolates the throttle. Under \Ktwo{}, \emph{larger budget}, a frame may run \harnessRoundsFast{} rounds instead of \harnessRoundsCap{}; by Proposition~\ref{prop:throttle} this stands for a faster machine only in the frames that run to the budget, so RQ3 measures the machine itself as well. Under \Kthree{}, \emph{renderer-less tool}, the redraw emulation is switched off, so the harness behaves like an execution of the virtual machine with no renderer attached, in which only \emph{wait} requests a redraw.

\runin{The check and its witness} The state after a frame is every variable of the stage and the sprites and, for every sprite and clone, hidden ones included, its position, direction, size, costume and visibility. For \Kone{} and \KoneM{} the hidden or muted sprites' own state is left out of both sides, so that the two rows of Table~\ref{tab:diverge} compare the same state; RQ4 reports \Kthree{} under the same projection as well. A project is \emph{rate-sensitive under a knob} when, for the same input and seed, its end state under the knob differs from its end state under the reference run. The property relates two runs of the same program, a hyperproperty~\cite{clarkson2008hyper}, and the end-state comparison is one of ScratchLens's observation lenses~\cite{si26lens}; the criterion is deliberately literal, and Section~\ref{sec:eval} reports which kinds of state moved and which a player could see. The check runs the reference and the knob, one child process each, and reports the first frame at which they differ, the variables and sprites that differ there, and the sprites it hid or muted. A project it calls \emph{unmoved} differs in no frame, so it is not rate-sensitive; the converse fails, since two runs can part and meet again by frame 300, and the sweeps count the end state only.

\runin{The lint} By Definition~\ref{def:frame} a loop is stepped more than once in a frame only when nothing requests a redraw and something is running. A loop whose body contains no primitive that can request a redraw (motion, looks, pen, a clone's creation) and no block that parks the thread therefore free-runs whenever nothing else requests, whatever its own sprite is doing. The exception is a loop inside a custom block that runs without screen refresh, where the warp timer paces it by a documented rule. The predicate follows a call to a custom block into the block's definition and skips the loops inside such blocks. \tool{} flags every loop that meets the condition, in the tradition of the pattern linters for Scratch~\cite{boe13hairball,fraser21litterbox,obermueller2021catnip}, by writing a comment next to it into a copy of the project file, which the unmodified editor shows (Figure~\ref{fig:lint}). RQ1 reports the flag's precision and recall against the runs. One might wonder why the comment does not simply tell the child to add a \emph{wait}. For the loops that write their count into a variable, the count is the program's meaning and a \emph{wait} would change it. Hence the comment carries the one sentence of Section~\ref{sec:cases}, as a question and without an instruction. A 300-frame run takes about two seconds on the laptop of Section~\ref{sec:study}.

\begin{figure}[t]\centering
\includegraphics[width=0.66\columnwidth]{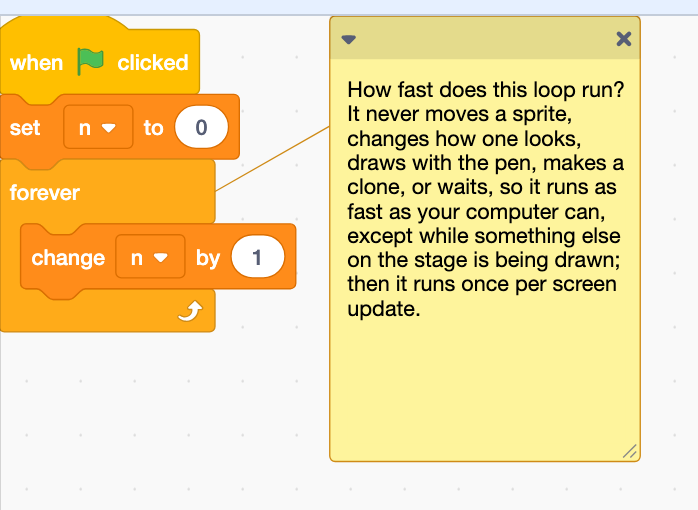}
\caption{The lint as the editor shows it (the script area of the public Scratch editor, with the lint's copy of the project loaded). \tool{} writes a comment next to every loop whose body can neither request a redraw nor park the thread, into a copy of the project file; the unmodified editor shows the comment when the copy is opened. The comment is the one sentence of Section~\ref{sec:cases}, as a question and with no instruction. Here the cat's counting loop of Figure~\ref{fig:example} is flagged; the apple's loop, which moves the sprite, is not.}
\label{fig:lint}
\end{figure}

\section{Evaluation}\label{sec:eval}

\subsection{Setup}\label{sec:study}
\runin{Research questions} \textbf{RQ1}: how much code can run at the speed of the screen, and does a static flag find it? \textbf{RQ2}: how often does the throttle change what a project does when the sprites that draw stop requesting redraws (\KoneM{}), how does hiding them differ (\Kone{}), and what changes? \textbf{RQ3}: how often does the machine change it (\Ktwo{}, and the editor's own runtime on a real clock), and is the emulated gate the editor's, frame by frame? \textbf{RQ4}: how often does a renderer-less execution (\Kthree{}) change it? \textbf{RQ5}: does any of this reach a grader's verdict?

\runin{Corpora} All populations are public. \emph{Games} are \nCorpusProjects{} public Scratch~3 projects fetched in September 2026 from the public search endpoint. We took the first 60 results in popularity order for each of twelve genre queries and the first 150 in trending order for \emph{game}, de-duplicated them, skipped projects over 15\,MB, and cut the list at 500. To see whether the numbers are the games' alone, a \emph{random sample} draws project ids uniformly between 1 and \randomMaxId{}, the largest id in September 2026, and keeps every id that names a shared Scratch~3 project with a downloadable file. That gave \nRandomOk{} projects from \nRandomAttempts{} ids (the rest named no shared project, were Scratch~2 files or failed to download), weighted by creation time and not by popularity. \emph{Tutorial remixes} are every public remix of two Scratch-team tutorials, \emph{Catch the Fish} (\nFlipRemixesFish{}) and \emph{Make It Fly} (\nFlipRemixesFly{}). A rule checker grades each remix, giving each rule three verdicts, pass, fail and untested (the rule's trigger never occurred, the vacuity of model checking~\cite{beer1997vacuity,kupferman2003vacuity}), and driving the trigger with probes a player could give. Because rules we wrote cannot show what rules written by others do, RQ5 adds \emph{Whisker's own example suites}, the \nWhiskerSuites{} projects and \nWhiskerTests{} tests that ship with Whisker~\cite{stahlbauer19}, and the \emph{Boat Race} tutorial of the Raspberry Pi Foundation~\cite{rpiboatrace}. Of the tutorial's 31,753 public remixes we fetched the first \nBoatRemixes{} in the repository's listing order that are Scratch~3 files.

\runin{Runs} A static pass classifies every loop: one with no block that can request a redraw and none that parks the thread is \emph{never-drawing} (the lint's predicate), and a never-drawing loop that changes a variable unconditionally on every iteration is a \emph{free-running counter}. We then run every game and every sampled project under the reference semantics and under each of the four knobs, in two runs each. \emph{Idle} is the green flag and nothing else. \emph{Played} holds, from the fifteenth frame on, the keys the project listens to (the right or left arrow, the up arrow or \emph{w}, and the space bar), so that \nDrivenWithKeys{} of the \nDrivenMeasured{} played games receive an input at all. Each run lasts 300 frames at the default budget of \harnessRoundsCap{} rounds per frame with seed 1, in its own child process on the laptop named below. A run \emph{free-runs} when it averages more than one round per frame over its 300 frames (\nIdleFree{} idle games do; \nIdleFreeAboveTwo{} average more than two rounds and \nIdleFreeAboveTen{} more than ten). A run is cut off by a step cap of five million block steps, a 150-second load-and-run watchdog, or an exception in the virtual machine. A project is excluded from a cell of Table~\ref{tab:diverge} when its reference run or that knob's run was cut off. In the final sweeps the reference run failed at load for \nGamesRefCut{} game and \nRandomRefCut{} sampled projects, all on an extension that needs a browser, and one sampled project's renderer-less run hit the step cap. The hidden and muted rows count the projects with a visible drawing sprite to hide or mute: \nGamesNoSprite{} games and \nRandomNoSprite{} sampled projects have none. For RQ3 the \nBudgetSubset{} idle games that free-run in the reference run are also run at budgets of 6, 99, 399 and 1,599 rounds per frame, with the witness at 99 against 24.

\runin{The editor on a real clock} The real-clock numbers come from the public \texttt{scratch-vm} (5.0.300) and \texttt{scratch-render} (2.2.84) bundles in headless Chromium 153 with nothing virtualised. Here \texttt{vm.start()} runs the real timers, and we read the state (every variable and, for the visible sprites, the rounded position, direction, size and costume) after ten seconds of wall-clock time on \editorMachineShort{}, with processor throttling from the browser's own emulation. The protocol repeats, as measurement studies recommend~\cite{mytkowicz2009wrong}: one warm-up run discarded, then three measured runs per game and rate, the project reloaded and \texttt{Math.random} re-seeded before each. The subset is \nRealProjects{} of the played games (every eighth of the id-sorted list, cut at \nRealProjects{}). A game is \emph{stable} at a rate when its three runs agree exactly, \emph{free-runs} at $1\times$ when some run steps a loop more than once per frame on average, and \emph{differs between rates} when every run at $20\times$ differs from every run at $1\times$.

\runin{The emulated gate against the editor's, frame by frame} A page hosts the editor's own bundles with the real renderer attached, installs the harness's seeded generator, virtual clock, sequencer clock, \emph{warp} count and knobs, and drives the frames itself. The two runtimes then differ only in the gate (real against emulated) and in sensing (pixels against boxes). \Kthree{} in the browser suppresses every request but \emph{wait}'s, as the harness does, and keeps the renderer for sensing. For each of the \nRealProjects{} games and each knob we run 300 frames in both, with the played keys, and compare every frame's number of rounds and its state; Table~\ref{tab:fidelity} defines the classes.

\runin{Rule verdicts} The remixes are checked under the reference semantics and under each knob with the checker's code untouched, since the knobs are the runtime's environment variables. Each remix gets its own tutorial's rules. The \emph{trigger-response} rules are the tutorials' instructions as \emph{whenever A, then B within k frames}: a clicked fish increases the score (cf1) and disappears (cf2) within five frames and a full score makes the game react (cf3). The up and down arrows move the cat within ten frames (mf1, mf2), and touching the buildings makes something happen within a second (mf3). The \emph{rate-reading} rules, written from the same instructions after the trigger-response results were in, read a value after a fixed time: the fish score after ten seconds of clicking (cfr1--cfr3), how far the buildings move in the first second (mfr1, mfr2), and how far a falling cat drops (mfr3). Whisker's suites run unchanged through a driver for Whisker's JavaScript API on the browser-hosted budgeted runtime with the real renderer, each test from a freshly loaded project. Whisker's own runner agrees with the driver on \nWhiskerNativeAgreeA{} of \nWhiskerNativeTests{} tests (Section~\ref{sec:threats}). The Boat Race suite is the tutorial's nine steps in the same API, with the tutorial's texts and our tolerances. Rules 1 to 3 read the boat's position and costume a tenth of a second after the flag, whether it follows the mouse, and its pace (15 to 150 pixels a second, the tutorial's one step per frame being 30). On touching the barrier, rules 4 to 6 expect a bubble, a costume change and, half a second later, the position back at the start; rules 7 to 9 expect a bubble on reaching the beach, the timer after three seconds, and the gate turning. By what the check reads, rules 1, 3, 6 and 8 read a position, a rate or a clock at a fixed time and rules 2, 4, 5, 7 and 9 check for an event within a window. The split was fixed when the suite was written. A rule whose trigger the driver cannot reach in a remix counts as untested. A \emph{flip} is a verdict under a knob that differs from the verdict under the reference, and a \emph{reversal} is a flip between pass and fail.

Every number below is a count over a population, regenerated from the result files by script, with no inferential statistics. Percentages are rounded to integers, except shares below one percent and shares over more than 2,000 frames or loops, which carry one decimal. The projects are public under the Scratch community licence; the study records project ids and no other user data, and no person took part in it.

\subsection{RQ1: How much code can run at the speed of the screen?}
Of \nLoops{} loops in the games, \nPureLoops{} (\pctPureLoops{}) contain no block that can request a redraw and no block that parks the thread: no motion, no looks, no pen, no clone, no wait, in the loop or in a custom block it calls. Of these, \nPureLoopsInWarp{} sit inside a custom block that runs without screen refresh, where the warp timer paces them by a documented rule, and are counted apart. The other \nPureLoopsOutsideWarp{} occur in \nProjectsPureLoopOutsideWarp{} projects (\pctProjectsPureLoopOutsideWarp{}), and in \nProjectsFreeLoopOutsideWarp{} projects (\pctProjectsFreeLoopOutsideWarp{}) at least one of them is a \emph{free-running counter}, so that the number of rounds per frame is written directly into the program's state. These loops are not exotic: by a keyword heuristic on their blocks and variable names, \nKindtimercounter{} are timers and counters, \nKindphysics{} compute speed, velocity or gravity, \nKindscorelives{} keep score or lives, \nKindlistwork{} work on lists, \nKindinputpolling{} poll input, \nKindrandomAI{} draw random numbers, and \nKindothercomputation{} do other computation. In the random sample, \nRandomProjectsPureLoopOutsideWarp{} of \nRandomProjects{} projects (\pctRandomProjectsPureLoopOutsideWarp{}) contain such a loop outside a warp block and \nRandomProjectsFreeLoopOutsideWarp{} (\pctRandomProjectsFreeLoopOutsideWarp{}) a free-running counter: the games are the busier population, and the phenomenon is not theirs alone.

The saved idle run says which of these loops are throttled. Of the \nIdleMeasured{} measured games, \nIdleRoundsOne{} step every loop once per frame for the whole ten seconds, \nIdleRoundsAtBudget{} run at the full budget throughout (a title screen waiting for a click), and the rest switch between the regimes. Does the static flag predict this? Against the idle run, the lint's confusion matrix is \nLintFlagFree{} flagged and free-running, \nLintFlagStill{} flagged and throttled, \nLintUnflagFree{} unflagged and free-running, \nLintUnflagStill{} neither. That is a recall of \pctLintRecall{} and a precision of \pctLintPrecision{}, against \pctIdleLintBaseline{} for flagging every project: at the level of projects the flag is no better than flagging everything. Per loop, where it is placed, we count a thread's steps, which the trace keys by the script's top block, so loops inside custom blocks are left out. A loop without a parking block iterates once per step, so it free-ran when its thread was stepped more than 300 times in the idle run. Of the \nPerLoopFlagged{} flagged loops outside custom blocks, \nPerLoopFlaggedFree{} free-ran (\pctPerLoopPrecision{}, and \pctPerLoopPrecisionEntered{} of the \nPerLoopFlaggedEntered{} that ran at all), against \pctPerLoopBase{} of all \nPerLoopNonParking{} loops without a parking block outside custom blocks; the flag finds \pctPerLoopRecall{} of the loops that free-ran. Both gaps have one cause, visibility is not in the code: a project free-runs without the flag when the only drawing block in its loops belongs to a hidden sprite or sits under a condition that stays false, and a flagged project stays throttled while any other visible sprite moves.

\takeaway{Answer to RQ1: more than half of the popular public games contain a loop whose speed is set from outside it, and about one in four writes that speed into a variable on every iteration; in the random sample the shares are \pctRandomProjectsPureLoopOutsideWarp{} and \pctRandomProjectsFreeLoopOutsideWarp{}. In the saved idle run, one game in ten runs every loop at the full budget for the whole ten seconds. The static flag finds \pctLintRecall{} of the projects that free-run and is right in \pctLintPrecision{} of the projects it flags, because visibility is not in the code.}

\begin{table}[b]\centering\footnotesize
\caption{RQ2--RQ4: projects whose state after ten seconds differs from the reference run's (\harnessRoundsCap{} rounds per frame, redraw gate emulated). Each cell is divergent / measured; the hidden and muted rows count the projects that have a visible drawing sprite.}
\label{tab:diverge}
\setlength{\tabcolsep}{4pt}\begin{tabular}{@{}lrr@{}}
\toprule
 & Idle & Played \\
\midrule
\multicolumn{3}{@{}l}{\emph{Games (\nCorpusProjects{})}} \\
Free-running (reference run) & \nIdleFree{} / \nIdleMeasured{} & \nDrivenFree{} / \nDrivenMeasured{} \\
Drawing sprites hidden (\Kone{}) & \nIdleHiddenDiverge{} / \nIdleHiddenDen{} (\pctIdleHiddenDiverge{}) & \nDrivenHiddenDiverge{} / \nDrivenHiddenDen{} (\pctDrivenHiddenDiverge{}) \\
Drawing sprites muted (\KoneM{}) & \nIdleMutedDiverge{} / \nIdleMutedDen{} (\pctIdleMutedDiverge{}) & \nDrivenMutedDiverge{} / \nDrivenMutedDen{} (\pctDrivenMutedDiverge{}) \\
Budget \harnessRoundsFast{} (\Ktwo{}) & \nIdleFasterDiverge{} / \nIdleFasterDen{} (\pctIdleFasterDiverge{}) & \nDrivenFasterDiverge{} / \nDrivenFasterDen{} (\pctDrivenFasterDiverge{}) \\
Renderer-less tool (\Kthree{}) & \nIdleRawDiverge{} / \nIdleRawDen{} (\pctIdleRawDiverge{}) & \nDrivenRawDiverge{} / \nDrivenRawDen{} (\pctDrivenRawDiverge{}) \\
\midrule
\multicolumn{3}{@{}l}{\emph{Random sample (\nRandomProjects{})}} \\
Free-running (reference run) & \nRandomIdleFree{} / \nRandomIdleMeasured{} & \nRandomDrivenFree{} / \nRandomDrivenMeasured{} \\
Drawing sprites hidden (\Kone{}) & \nRandomIdleBDiv{} / \nRandomIdleBDen{} (\pctRandomIdleBDiv{}) & \nRandomDrivenBDiv{} / \nRandomDrivenBDen{} (\pctRandomDrivenBDiv{}) \\
Drawing sprites muted (\KoneM{}) & \nRandomIdleMDiv{} / \nRandomIdleMDen{} (\pctRandomIdleMDiv{}) & \nRandomDrivenMDiv{} / \nRandomDrivenMDen{} (\pctRandomDrivenMDiv{}) \\
Budget \harnessRoundsFast{} (\Ktwo{}) & \nRandomIdleCDiv{} / \nRandomIdleCDen{} (\pctRandomIdleCDiv{}) & \nRandomDrivenCDiv{} / \nRandomDrivenCDen{} (\pctRandomDrivenCDiv{}) \\
Renderer-less tool (\Kthree{}) & \nRandomIdleRDiv{} / \nRandomIdleRDen{} (\pctRandomIdleRDiv{}) & \nRandomDrivenRDiv{} / \nRandomDrivenRDen{} (\pctRandomDrivenRDiv{}) \\
\bottomrule
\end{tabular}
\end{table}

\subsection{RQ2: Stop the drawing, change the game}
Table~\ref{tab:diverge} counts the projects whose state after ten seconds differs from the reference run's. Under \KoneM{}, the throttle alone, \nIdleMutedDiverge{} of the \nIdleMutedDen{} idle games that have a drawing sprite end somewhere else (\pctIdleMutedDiverge{}), and \nDrivenMutedDiverge{} of \nDrivenMutedDen{} played games do (\pctDrivenMutedDiverge{}); in the random sample the shares are \pctRandomIdleMDiv{} and \pctRandomDrivenMDiv{} (95\% interval \ciRandomDrivenMutedDiverge{} for the played share). What moves is what a player keeps score of: of the \nDrivenMutedDiverge{}, \nDrivenMutedStage{} differ in what the stage shows, and the variable that moves most often is a score (\movedMScore{} games), with clone bookkeeping, positions and timers after it.

Hiding the same sprites (\Kone{}) changes fewer games, \nIdleHiddenDiverge{} (\pctIdleHiddenDiverge{}) idle and \nDrivenHiddenDiverge{} (\pctDrivenHiddenDiverge{}) played, and the surprise is that it changes different ones. In the played runs \nDrivenHiddenAndMuted{} games diverge under both knobs, \nDrivenMutedNotHidden{} under muting only and \nDrivenHiddenNotMuted{} under hiding only. Muting leaves the sprite in play, so an enemy that moves twenty times faster still catches the player; hiding takes it out of play, so a released enemy no longer touches anything. In all \nDrivenHiddenNotMuted{} games that diverge under hiding alone the two knobs removed the same requests, so what separates them is the second mechanism, the ordinary semantics of a hidden sprite (Section~\ref{sec:cases}'s control).

Three checks say the number is not an artefact of the set-up. Play matters, but not the way the keys were held: of the games that receive a key, \nDrivenKeysMutedDiverge{} of \nDrivenKeysMutedDen{} (\pctDrivenKeysMutedDiverge{}) diverge under muting, against \nDrivenNoKeysMutedDiverge{} of \nDrivenNoKeysMutedDen{} (\pctDrivenNoKeysMutedDiverge{}) of those that receive none. With the keys tapped instead of held (down for six frames, up for 24, from frame 15) the muted and renderer-less shares are \pctTappedMutedDiverge{} and \pctTappedRawDiverge{} (\nTappedMutedDiverge{} of \nTappedMutedDen{}, \nTappedRawDiverge{} of \nTappedRawDen{}). The seed matters to the state but not to the label: a change of seed alone moves \pctSeedTwoDiverge{} of the played games. Yet at a second seed the muted row is \nSeedTwoMutedDiverge{} of \nSeedTwoMutedDen{} (\pctSeedTwoMutedDiverge{}) and the per-game label agrees across the two seeds in \nSeedTwoMutedLabelAgree{} of \nSeedTwoMutedPaired{} games (hidden \nSeedTwoHiddenLabelAgree{} of \nSeedTwoHiddenPaired{}, renderer-less \nSeedTwoRawLabelAgree{} of \nSeedTwoRawPaired{}). And the knob mutes every drawing sprite at once; muting one at a time, the scenario of the introduction, moves \nSingleGamesAny{} of \nSingleGames{} games (\pctSingleGamesAny{}) for at least one sprite, \nSingleSpritesDiverge{} of the \nSingleSprites{} drawing sprites in all. E1 fires when any sprite requests, so one sprite's silence releases a loop only in the frames where no other sprite draws.

\takeaway{Answer to RQ2: with the drawing sprites' requests muted, about one played public game in \mutedOneIn{} ends its first ten seconds in a different state, \pctSingleGamesAny{} when one sprite at a time is muted, and what changes is score, position, clones and time. Hiding the same sprites changes fewer games, because a hidden sprite also leaves the game, and because a script often shows it again.}

\begin{table}[t]\centering\footnotesize
\caption{RQ3, the emulated gate against the editor's on the \nFidAGames{} real-clock games, 300 frames each, the harness against the editor's bundles with the harness's substitutions installed (real renderer), under the reference and each knob (\Kone{} hidden, \KoneM{} muted, \Ktwo{} budget 99, \Kthree{} renderer-less). A game is \emph{identical} in all its frames, or its first difference is in the rounds (\emph{gate}), in rounds and state at once (\emph{tie}), or in the state alone (\emph{sensing}: pixels against boxes). \emph{Rounds equal} is the share of the 18,000 frames with the same number of rounds; \emph{same status}, the games both call free-running; \emph{same verdict}, the games both agree on whether the knob moves the end state (the game without a drawing sprite left out).}
\label{tab:fidelity}
\setlength{\tabcolsep}{5pt}
\begin{tabular}{@{}lrrrrr@{}}
\toprule
 & Ref. & \Kone{} & \KoneM{} & \Ktwo{} & \Kthree{} \\
\midrule
Identical (games) & \nFidAIdentical{} & \nFidBIdentical{} & \nFidMIdentical{} & \nFidCIdentical{} & \nFidRIdentical{} \\
Gate & \nFidAGate{} & \nFidBGate{} & \nFidMGate{} & \nFidCGate{} & \nFidRGate{} \\
Tie & \nFidAAmbiguous{} & \nFidBAmbiguous{} & \nFidMAmbiguous{} & \nFidCAmbiguous{} & \nFidRAmbiguous{} \\
Sensing & \nFidASensing{} & \nFidBSensing{} & \nFidMSensing{} & \nFidCSensing{} & \nFidRSensing{} \\
\addlinespace[3pt]
Rounds equal (frames) & \pctFidARoundsAgree{} & \pctFidBRoundsAgree{} & \pctFidMRoundsAgree{} & \pctFidCRoundsAgree{} & \pctFidRRoundsAgree{} \\
Same status (games) & \nFidAClassAgree{}/\nFidAGames{} & \nFidBClassAgree{}/\nFidBGames{} & \nFidMClassAgree{}/\nFidMGames{} & \nFidCClassAgree{}/\nFidCGames{} & \nFidRClassAgree{}/\nFidRGames{} \\
Same verdict (games) & -- & \nWebAgreeB{}/\nWebDenB{} & \nWebAgreeM{}/\nWebDenM{} & -- & \nWebAgreeR{}/\nWebDenR{} \\
\bottomrule
\end{tabular}
\end{table}

\subsection{RQ3: Same game, faster computer}
Under \Ktwo{} we raise the budget from \harnessRoundsCap{} to \harnessRoundsFast{} rounds per frame, which is what a faster machine does to exit E3 in the frames that reach it. The end state changes in \nIdleFasterDiverge{} of \nIdleFasterDen{} idle games (\pctIdleFasterDiverge{}) and \nDrivenFasterDiverge{} of \nDrivenFasterDen{} played games (\pctDrivenFasterDiverge{}), and as rarely in the random sample (Table~\ref{tab:diverge}). The count is small, and the mechanism says why (Proposition~\ref{prop:throttle}): the budget matters only to a loop that is free-running, and a loop free-runs only in frames where nothing visible draws, which in a game being played is rarely the case. The \nBudgetSubset{} idle games that do free-run say the same at 6, 99, 399 and 1,599 rounds per frame: the end state differs from the default in \nBudgetDivergeSix{}, \nBudgetDivergeNinetyNine{}, \nBudgetDivergeFifteenNinetyNine{} and \nBudgetDivergeFifteenNinetyNine{} games, and the witness at 99 against 24 finds \nBudgetWitnessNever{} of the \nBudgetWitness{} differing in no frame at all.

Nor are the two knobs that release loops budget-bound. With both runs at 99, 399 and 1,599 rounds (a renderer-less run the step cap cuts off left out, \nDrivenBudgetFifteenNinetyNineRawExcluded{} at 1,599), muting changes \pctDrivenBudgetNinetyNineMutedDiverge{}, \pctDrivenBudgetThreeNinetyNineMutedDiverge{} and \pctDrivenBudgetFifteenNinetyNineMutedDiverge{} of the played games and the renderer-less run \pctDrivenBudgetNinetyNineRawDiverge{}, \pctDrivenBudgetThreeNinetyNineRawDiverge{} and \pctDrivenBudgetFifteenNinetyNineRawDiverge{}, against \pctDrivenMutedDiverge{} and \pctDrivenRawDiverge{} at 24, and the reference runs at 399 and 1,599 reach the same end state in \nBudgetRefSame{} of \nBudgetRefPairs{} games. A run repeated at the same seed reproduces its end state and its rounds exactly (\nRepeatIdentical{} of \nRepeatPairs{}). The aligned build run on the \nFidAGames{} games under WebKit and Firefox runs the same rounds as under Chromium in every frame and, with no renderer, reaches Chromium's end state in \nEngineNWebKitSame{} of \nEngineNWebKitPairs{} games under both (with the renderer, whose pixels the browsers draw differently, in \nEngineWebKitSame{} and \nEngineFirefoxSame{}). On a second machine (a Windows workstation, x86-64, Node~24 against Node~22) the played sweep's reference, muted and renderer-less runs reach the laptop's end state and rounds in \nRemoteReferenceSame{}, \nRemoteMutedSame{} and \nRemoteRawSame{} of \nRemoteReferencePairs{} games, the exception a camera coordinate's sixteenth digit, the engine's arithmetic that the proposition excepts.

\runin{The editor on a real clock} The real clock says where the machine's influence comes from, and it is mostly not the throttle. Of \nRealProjects{} played games, \nRealAgreeOne{} (\pctRealAgreeOne{}) are \emph{stable} at $1\times$, three runs reaching exactly the same state with the random source pinned; at $20\times$ it is \nRealAgreeTwenty{}, and at both rates \nRealAgreeBoth{}. On the machine throttled twenty times, \nRealDifferTwenty{} games end somewhere else in every one of their runs, and only \nRealDifferTwentyFree{} of them free-run at $1\times$. The other \nRealDifferTwentyThrottled{} do not, and \nRealThrottledNoFreeFrame{} of them have no frame with more than one round at either rate: their loops run once per frame on both machines, and what differs is what the wall clock does to their \emph{timer}, \emph{wait} and \emph{glide} blocks when each frame takes longer. A machine reaches a program in two ways, then, through the free-run count and through the program's own clock, and only the first is the throttle. Muting at native speed says what the harness says. On the \nRcmGames{} games run again with the drawing sprites' requests suppressed, three runs each, every muted run ends elsewhere than every reference run in \nRcmDiffersAll{} games, the released loops running \medRcmMutedRoundsPerFrame{} rounds per frame in the median run. The harness's muted label at 24 rounds agrees with that outcome in \nRcmAgree{} of \nRcmPaired{} games; in \nRcmWebOnly{} the native-speed runs part where 24 rounds did not, and in \nRcmHarnessOnly{} the harness parts where the native-speed runs end alike.

\runin{The emulated gate against the editor's, frame by frame} Table~\ref{tab:fidelity} is the check the aggregate could not make: the same executions compared after every frame. Under the reference, \nFidAIdentical{} games are identical in every one of their 300 frames, in rounds and in state. In \nFidAGate{} game the first difference is in the rounds, by one round in one frame, and the states never differ. In \nFidAAmbiguous{} game rounds and state first differ in the same frame, which either cause can produce. The other \nFidASensing{} differ first in the state with the rounds equal, where pixels answered a touch or a fence differently from boxes. Over all frames the two runtimes ran the same number of rounds in \pctFidARoundsAgree{}. The frames before the states part, where the comparison is between the same executions, are \nFidAPrefixFrames{} of the 18,000 (the games that part do so early, at frame \medFidAFirstState{} in the median), and in them the rounds differ in \nFidAPrefixDisagree{} frame, the gate case. The gate is exercised in that prefix: \nFidPrefixMulti{} of its frames ran more than one round in one runtime or both, in \nFidPrefixMultiGames{} games, and the two ran the same number in \nFidPrefixMultiAgree{} of them; over the 300 frames, \nFidFreeGames{} of the \nFidAGames{} games free-run in one runtime or both, and the two agree on which in \nFidFreeGamesAgree{}. Under the knobs the picture is the same, with the rounds equal in \pctFidRRoundsAgree{} of frames at the least, and the verdicts transfer. Run in the editor's bundles with the real renderer and pixel sensing, the same games diverge under muting in \nWebDivergeM{} of \nWebDenM{} cases (\nWebBothM{} in the harness on the same games). The two runtimes agree on which games diverge in \nWebAgreeM{} of \nWebDenM{} under muting, \nWebAgreeB{} of \nWebDenB{} under hiding and \nWebAgreeR{} of \nWebDenR{} without a renderer, every disagreement a game the editor's bundles move and the harness does not. The check also found errors: its first runs found four defects in our reading of the source, which the emulation had implemented and the corpus would have hidden (the pen's and a visible clone's creation requests missing, the layer methods' and a hidden bubble's requests spurious). Every number in this paper comes from the corrected table and harness.

\takeaway{Answer to RQ3: the machine reaches a loop through the free-running frames, so within ten seconds a larger budget changes few games, played or idle; the frame rate reaches the throttled loops instead, through the wall clock in their own blocks. Frame by frame, the emulated gate runs the same rounds as the editor's in \pctFidARoundsAgree{} of frames under the reference and at least \pctFidRRoundsAgree{} under every knob, and the two runtimes agree on which games a knob moves.}

\subsection{RQ4: Tools that run the wrong program}
Under \Kthree{} the redraw emulation is off: the situation of a tool that executes the virtual machine without a renderer and without emulating the gate, in which a frame that no \emph{wait} interrupts runs to the budget. It is the largest effect of the study. Of the idle games, \nIdleRawDiverge{} of \nIdleRawDen{} (\pctIdleRawDiverge{}) end somewhere else than under the emulated gate, and \nDrivenRawDiverge{} of \nDrivenRawDen{} played games (\pctDrivenRawDiverge{}) do; in the random sample the shares are \pctRandomIdleRDiv{} and \pctRandomDrivenRDiv{} (Table~\ref{tab:diverge}). The knob removes the throttle from every loop at once, and under the projection of the hidden and muted rows, the drawing sprites left out, it still moves \nDrivenProjRawDiverge{} of \nDrivenProjRawDen{} (\pctDrivenProjRawDiverge{}). What moves is visible: of the \nDrivenRawDiverge{} divergent played games, \nDrivenRawStage{} differ in what the stage shows (\nDrivenRawStageOnly{} on the stage only) and \nDrivenRawVarsOnly{} only in variables, and the variable that moves most often is again a score (\movedRScore{} games), then positions, timers and velocities. In the catch-the-fruit game 120830462 the score ends at \caseFruitScoreR{} instead of \caseFruitScoreA{}, because the fruit falls at \caseFruitRoundsR{} rounds per frame: a grader that runs this game without a renderer grades a game the editor never runs.

\takeaway{Answer to RQ4: executed without a renderer and without the gate, more than half of the public games we played end their first ten seconds somewhere else, most of them visibly on the stage. A tool that runs the virtual machine without a renderer must emulate the redraw gate or it runs a different program.}

\subsection{RQ5: Does it reach a verdict?}
A rubric's fate depends on what it reads. Under the reference the trigger-response rules give \nFlipBasePass{} pass, \nFlipBaseFail{} fail and \nFlipBaseUntested{} untested verdicts, and the rate-reading rules \nFlipRateBasePass{}, \nFlipRateBaseFail{} and \nFlipRateBaseUntested{}; a flip that is not a reversal is to or from untested. Of the \nFlipRules{} trigger-response verdicts, \nFlipB{} (\pctFlipB{}) flip when the drawing sprites are hidden, almost all to untested from pass and from fail alike, because a hidden fish cannot be clicked and a hidden cat cannot touch: the second mechanism of RQ2, not the throttle. The throttle's own share is the muted row, \nFlipM{} flips with \nFlipRevM{} reversals (\ruleRevsM{}); the budget flips \nFlipC{} and the renderer \nFlipR{} (\nFlipRevR{} reversals).

The rate-reading rules are the ones the mechanism says a throttle should move, and they move in both directions. Of \nFlipRateRules{} verdicts, \nFlipRateM{} flip under muting and every one is a reversal (\flipRateKindsM{}; by rule, \ruleRateRevsM{}); without a renderer \nFlipRateR{} flip, \nFlipRateRevR{} of them reversals; under hiding \nFlipRateB{} flip, almost all to untested, because a rule about a hidden sprite's position is blocked rather than refuted; under the larger budget \nFlipRateC{} flip, for the reason RQ3 gave, since in the reference run these loops are paced by the screen and a larger budget never reaches them. A fail$\rightarrow$pass reversal is a pass that the reference semantics does not give: a false accept whenever the rubric means the editor's pace, which the tutorial's does. The \nFlipRateM{} flips are few because most remixes keep the tutorial's drawing loops.

\runin{Rules written by others} Rules we did not write repeat the question. Whisker's \nWhiskerTests{} tests on \nWhiskerSuites{} projects give, under the reference, \nWhiskerBasePass{} pass, \nWhiskerBaseFail{} fail and \nWhiskerBaseSkip{} skip (Whisker's own third outcome), as the tests and the programs ship. Hiding the drawing sprites flips \nWhiskerFlipsB{} (\nWhiskerRevB{} reversals). Muting them flips \nWhiskerFlipsM{} verdicts (\pctWhiskerFlipsM{} of the tests, \nWhiskerFlipsMAgreed{} of the \nWhiskerAgreedTests{} on which Whisker's own runner and our driver agree), \nWhiskerRevM{} of them reversals; of the 20, 18 read a position, a count or the timer after a fixed time, and 2 check a touch. The renderer-less run flips \nWhiskerFlipsR{} (\pctWhiskerFlipsR{}) with \nWhiskerRevR{} reversals, \nWhiskerFlipsRAgreed{} of the flips among the agreed tests, and the larger budget flips \nWhiskerFlipsC{}. The Boat Race tutorial tells the same story from a classroom's side, with larger numbers. Of \nBoatVerdicts{} verdicts, \nBoatFlipsM{} (\pctBoatFlipsM{}) flip under muting and \nBoatFlipsR{} (\pctBoatFlipsR{}) without a renderer, \nBoatFlipsC{} under the larger budget (hiding also flips \nBoatFlipsB{}, with \nBoatRevB{} reversals, the same total by coincidence); \nBoatRevM{} of the muted flips and \nBoatRevR{} of the renderer-less ones are reversals. Split by kind, the rules say what ours said: under muting the four rules that read a rate, a position or a clock after a fixed time carry \nBoatRateFlipsM{} of the \nBoatFlipsM{} flips, \nBoatRateRevM{} of them reversals. The five rules that check for an event within a window carry \nBoatTriggerFlipsM{}, with \nBoatWindowRevM{} reversals (that the reversals also total \nBoatRateFlipsM{} is a coincidence). Why do these rate rules flip in a fifth of the remixes when ours flip in a few? They read the pace and position of the boat itself, the sprite the knob mutes, whose motion loop is the released loop; our rate rules read counters kept by loops that the tutorials' other drawing sprites keep throttled. The boat's pace rule alone flips \nBoatSpeedFlipsM{} times under muting, since a boat that covers 15 to 150 pixels a second at one step per frame covers hundreds once its loop free-runs.

\takeaway{Answer to RQ5: rules that check for an event within a window of frames rarely reverse under the throttle, the budget and the renderer (\nFlipRevM{} and \nFlipRevR{} reversals in \nFlipRules{} verdicts), and flip under hiding because a hidden sprite cannot be clicked or touched. Rules that read a value after a fixed time reverse under muting and under the renderer, in both directions, and not under the budget, which reaches only loops that already free-run. Whisker's own suites and a tutorial's rules flip in the same places.}

\section{Two Cases, One Control, and What To Do}\label{sec:cases}
The counts of Section~\ref{sec:eval} hide two mechanisms, so we walk two real games through the knobs, and then say what each audience can do.

\runin{Case 1: a pong game whose score lives in a loop that never draws} Project 974172893 (\emph{The great PONG game}, played with no keys) has two sprites, a ball and a paddle, and one loop that keeps score: \emph{forever, if touching Paddle then change score by 1}. The loop never draws. The ball's other script moves it, and that motion ends every frame, so as saved the score loop runs \casePongTwoRoundsA{} round per frame, a touch that lasts a frame counts once, and the score after ten seconds is \casePongTwoScoreA{}. Mute the ball's and the paddle's requests (\KoneM{}) and the same loop runs \casePongTwoRoundsM{} rounds per frame, counting every round in which the ball, which now also moves \casePongTwoRoundsM{} times a frame, touches the paddle; the score is \casePongTwoScoreM{}. Hide the same two sprites (\Kone{}) and the loop runs as fast, but a hidden ball touches nothing and the score stays at \casePongTwoScoreB{}: the two mechanisms of RQ2 in one game. In the editor's bundles with the real renderer and pixel sensing, the reference, the muted and the hidden runs score \casePongTwoWebScoreA{}, \casePongTwoWebScoreM{} and \casePongTwoWebScoreB{}. Project 1051219841, a pong game with a stopwatch, is the pen exception. Hiding or muting the sprites that draw leaves every loop at \casePongRoundsB{} round per frame, because a hidden sprite draws the lava with the pen, which requests whether or not its sprite is visible (Table~\ref{tab:draw}).

\runin{The control: a runner whose lives depend on a hidden sprite for another reason} Project 1273597743 is an endless runner in which a toucan dodges buildings and eats watermelons. As saved, the player loses all three lives within ten seconds of an idle run (LIVES reaches \caseRunnerLivesA{}). With the \caseRunnerHidden{} drawing sprites hidden (\Kone{}) the loops run \caseRunnerRoundsB{} rounds per frame and LIVES ends at \caseRunnerLivesB{}, the starting value. This is not the throttle: a hidden sprite neither touches nor is touched, so the buildings never hit the toucan. Under \KoneM{}, where the same sprites stay visible and touchable and only their requests are gone, the loops run \caseRunnerRoundsM{} rounds per frame (fewer than under hiding, because the game ends when the lives run out and its loops stop) and LIVES ends at \caseRunnerLivesM{}. The editor's bundles agree (\caseRunnerWebLivesA{}, \caseRunnerWebLivesB{} and \caseRunnerWebLivesM{} lives under the reference, hidden and muted). Hence the throttle alone changes the pace of the game and not the number of lives lost.

\runin{For learners} We propose one sentence in place of the one the wiki gives, and we make no claim yet that children understand it. It reads: \emph{\theSentence{}} ``Changes how one looks'' stands for the looks blocks; the sentence's list is the lint's predicate. It describes the default mode: turbo mode and \emph{run without screen refresh}, which the Wiki documents, take a loop out of it. A teacher without a harness can compare two children's timers or scores only after checking that something visible moves in every frame of each project.

\runin{For automated graders} A grader that runs a program in a different environment from the editor's sees a different ten-second state for \pctDrivenRawDiverge{} of the games of RQ4, and whether a verdict follows depends on what the rubric reads (RQ5). The remedy is the budgeted semantics: emulate the gate, answer sensing from the renderer where the rubric reads a touch, and publish the budget with the verdict. A budget is a count of rounds, not a machine, so it makes a verdict reproducible rather than faithful to any one computer. RQ3 says that between 99 and 1,599 rounds the muted and renderer-less contrasts and the reference's end states barely move, so a grader can state a budget in that range and expect another grader's to agree. Where a verdict depends on a rate, the check should run alongside it: a project that diverges under muting is reported as rate-sensitive with its witness, a third outcome beside pass and fail, for the rubric's author to settle. On this corpus that is about one played game in \mutedOneIn{}, and a rule that checks for an event within a window is the rule that keeps it off that pile. On the \nCtRemixes{} tutorial remixes under our window rules, the check on the full state (the muted sprites' own included, since a rubric may read it) calls \nCtMoved{} (\pctCtMoved{}) rate-sensitive under muting; \nCtFlipped{} have a verdict that flips, \nCtBoth{} of them flagged. The \nCtMovedNotFlipped{} flagged without a flip are the price of the third outcome under window rules, and the \nCtFlippedNotMoved{} flipped without a flag are flips in probe runs the check did not repeat: the check belongs beside the verdict's own run.

\runin{For the Scratch platform} The throttle is a performance heuristic that leaked into program meaning: the redraw flag exists so that the screen is not drawn more often than it changes. Pacing every loop at one round per frame unless the project opts out, as \emph{run without screen refresh} already lets a custom block do, removes the throttle's share of \Kone{} and \Kthree{}. The cost is the speed that non-drawing loops enjoy today and the meaning of the \nFreeLoops{} free-running counters, which is why it needs the opt-out. A round budget in place of the wall-clock exit removes the scheduler's share of \Ktwo{}.

\section{Limitations and Threats to Validity}\label{sec:threats}
\runin{The harness is not the editor} The harness runs \harnessRoundsCap{} rounds per frame where the editor on a fast laptop manages tens of thousands; RQ3's contrasts at 99 to 1,599 rounds and the native-speed muting run bound the difference, and we make no claim that Table~\ref{tab:diverge}'s percentages are bounds at the editor's thousands. Sensing is answered by bounding boxes rather than pixels, the colour predicates by \emph{false}, and a \emph{warp} block yields after 500 checks; the answers are the same under every knob, so a knob's divergence is not caused by them. A colour predicate appears in \nGamesColour{} games (\pctGamesColour{}) and \nRandomColour{} sampled projects; on the games without one the muted and renderer-less shares are \pctDrivenNoColourMutedDiverge{} and \pctDrivenNoColourRawDiverge{} (\nDrivenNoColourMutedDiverge{} of \nDrivenNoColourMutedDen{}, \nDrivenNoColourRawDiverge{} of \nDrivenNoColourRawDen{}). The frame-by-frame check of RQ3 aligned the gate but cannot align sensing: \nFidASensing{} of the \nFidAGames{} games part where pixels and boxes answer a touch or a fence differently. The percentages of Table~\ref{tab:diverge} are therefore differences under the editor's gate and the harness's sensing, on \nFidAGames{} games checked and \nCorpusProjects{} not.

\runin{Ten seconds, held keys, a literal comparison} Our runs are short and our player is crude: it holds or taps the keys a project listens to and never clicks. A longer horizon and a driver that plays the game would reach other states, and either can raise or lower the counts of RQ2 and RQ4. A rate change also changes how fast a free-running loop consumes the random stream, so a project that draws random numbers in such a loop can diverge for that reason as well as for the throttle's. Such a loop is in \nProjectsRandomPureLoop{} games, among them \nDrivenMutedDivergeRandomLoop{} of the \nDrivenMutedDiverge{} that diverge under muting.

\runin{Rules, remixes and the Whisker driver} We wrote the rules of RQ5's first two sets from the tutorials' instructions and classified every rule by kind ourselves, with no second rater; the Boat Race split was fixed before its flips were counted. The Boat Race rules are the tutorial's steps in our tolerances. Whisker's own runner agrees with the driver on \nWhiskerNativeAgreeA{} of \nWhiskerNativeTests{} tests, and on \nWhiskerNativeAgreeW{} at a budget of 24,999 rounds. The disagreements are therefore not the budget but the drivers' timing of mouse moves, clone clicks and bubbles, which differs by a frame or two.

\runin{Populations} The games are popular projects found by genre queries, which favours animated projects with many sprites; the random sample holds stories and animations as well, and its divergence rows are lower, the direction the sampling frame predicted. Its free-running row is higher (\nRandomIdleFree{} of \nRandomIdleMeasured{} idle): a story in which nothing visible moves runs its loops to the budget, and what they write is rarely something a knob changes. The lint's per-loop precision (RQ1) is low, and the categories of RQ1's loops are a keyword heuristic whose residual class is the largest. We read the rule from \texttt{scratch-vm} \vmVersionNote{}; the work loop has had this shape since 2019~\cite{scratchvm2138}.

\section{Related Work}\label{sec:related}
\runin{Analyses of Scratch programs} Studies of the public repository have measured which blocks children use and how projects grow~\cite{aivaloglou2016kids}, which code smells recur~\cite{hermans2016smells,techapalokul2017quality} and which bugs are common~\cite{fraedrich2020bugs}. Static analyses look for patterns in the blocks (Hairball~\cite{boe13hairball}, Dr.\ Scratch~\cite{morenoleon15drscratch}, LitterBox~\cite{fraser21litterbox}, Catnip~\cite{obermueller2021catnip}).

\runin{Running Scratch programs, with and without a renderer} Every dynamic tool for Scratch runs the program somewhere. The editor and the player run \texttt{scratch-vm} in the browser with the renderer at 30 frames per second, and TurboWarp runs a fork at 30 or 60 whose turbo mode disables E1~\cite{turbowarpfps}. Itch translates the project to Python and runs the translation~\cite{johnson2016itch}; Bastet interprets a translation of it in a semantics of its own, LeILa~\cite{stahlbauer2020verified}; Whisker sends user events to a project running in Firefox or headless Chromium with the renderer and checks properties with test scripts~\cite{stahlbauer19,deiner23whisker,goetz2022model}; block-based tests run inside the editor under a five-second timeout~\cite{feldmeier24}; NuzzleBug is a record-and-replay debugger in the editor~\cite{deiner2024nuzzlebug}; SchedCheck and ScratchLens run \texttt{scratch-vm} without a renderer~\cite{schedcheck26,si26lens}. The tools that run in a browser with the renderer run the editor's gate and its wall clock, and what they inherit is \Ktwo{}. Whisker's wrapper pauses one period divided by the acceleration $a$ between test steps, so a free-running frame still takes its 25\,ms of real work while the clock the program reads runs $a$ times faster. Its documentation recommends $a \leq 10$ ``as very low execution times may lead to non-deterministic program behaviour'', and the mechanism predicts that symptom. The tools that execute \texttt{scratch-vm} without a renderer run under \Kthree{} unless they emulate the gate, and none of their descriptions says whether a frame ends at a redraw request. Itch and Bastet have no frame at all, so a loop's rate there is the translation's. Automated assessment of programs is a mature practice~\cite{ihantola2010review,paiva2022assessment}, and feedback generators compare a student's traces with a reference's~\cite{singh2013feedback,gulwani2018clustering}. These systems assume that a program's behaviour is fixed by its text and its input, which for a Scratch program holds only under a stated budget, an emulated gate and aligned reads (Proposition~\ref{prop:agree}).

\runin{Hidden semantics of event-driven programs} The order in which event handlers run is a hidden coordinate of event-driven programs, and race detectors for web applications find the schedules that expose it~\cite{petrov2012race,raychev2013eventracer,desai2015pingpong}. For Scratch, SchedCheck showed that the order in which cooperative threads run within a frame follows the sprites' stacking order and that about a fifth of concurrent student projects are sensitive to it~\cite{schedcheck26}. That work holds the frame fixed and varies the order; this paper holds the order fixed and varies the rate, and the two checks compose. SchedCheck and ScratchLens run the virtual machine without a renderer and so inherit \Kthree{}; their order-sensitivity counts were taken under the free-running frame, and repeating them under the budgeted semantics is future work.

\runin{Determinism, flakiness and metamorphic tests} Deterministic replay and multithreading replace a decision the clock makes by one the program's progress makes~\cite{choi1998dejavu,olszewski2009kendo,bergan2010coredet,liu2011dthreads}, or choose schedules systematically or at random under a bound~\cite{musuvathi2008chess,burckhardt2010pct}. The budgeted semantics is the same move for one decision, the number of rounds in a frame. Timing is the most common cause of flaky tests~\cite{luo14flaky,eck2019flaky,parry2021flaky,lam2019idflakies,gligoric2015nondex}, and record-and-replay removes the wall clock from a web application's execution~\cite{andrica2011warr,burg2013replay}, as the budgeted semantics does for Scratch's frame. The check of Section~\ref{sec:tool} is a metamorphic test~\cite{segura16metamorphic,chen2018metamorphic}: it has no oracle for what a project should do and asks whether a transformation that should not matter changes the output. The throttle is a source of flakiness with a mechanism, not a race between two threads but a gate that couples all of them to the screen.

\runin{Learners' models of execution} The education literature has measured what novices believe about execution: misconceptions of loops and variables in block-based programming~\cite{grover2017misconceptions,swidan2018misconceptions}, the gap between what children's Scratch code does and what they think it does~\cite{salac2020build}, and preconceptions of concurrency~\cite{kolikant2001gardeners,meerbaum2013learning}. The sentence and the lint of Section~\ref{sec:cases} are proposals in this space, not results.

\section{Conclusion}\label{sec:conclusion}
A Scratch loop that draws runs once per frame; a loop that neither draws nor waits runs as fast as the computer allows, until something visible on the stage asks the screen to redraw. In the editor, hiding one sprite makes another count \editorHiddenOverVisibleThirty{} times higher. On \nCorpusProjects{} popular public games, more than half contain a loop paced this way, muting the sprites that draw changes the state after ten seconds of about one played game in \mutedOneIn{}, a third of them on the stage, and executing without a renderer changes more than half of them. What the throttle moves is a score, a clock or a position, so a rubric that reads one of those is the rubric at risk, and rules of the form \emph{whenever A, then B within k frames} are not. A tool that runs the virtual machine without a renderer must emulate the gate and state its budget, or it grades a program the editor never runs. The population under this rule is 135 million registered users, most of them children; the sentence they could be taught is in Section~\ref{sec:cases}.

\bibliographystyle{ACM-Reference-Format}
\bibliography{refs}
\end{document}